\documentclass[12pt]{article}%
\usepackage[letterpaper, top=1in, bottom=1in, left=1in, right=1in]{geometry}
\usepackage{amsfonts, amsmath, amsthm, amssymb}
\usepackage{enumitem, graphicx, stmaryrd, color, bbm}
\usepackage{url, algorithm}
\usepackage{float}
\floatstyle{ruled}
\usepackage{dsfont}
\usepackage{footmisc}
\usepackage[T1]{fontenc}

\usepackage[pagebackref=true,colorlinks]{hyperref}
	\hypersetup{linkcolor=[rgb]{.7,0,.7}}
	\hypersetup{citecolor=[rgb]{.5,0,.5}}
	\hypersetup{urlcolor=[rgb]{.7,0,.7}}

\newtheorem{theorem}{Theorem}[section] 
\newtheorem{lemma}[theorem]{Lemma}

\newtheorem{definition}[theorem]{Definition}

\newtheorem{corollary}[theorem]{Corollary}
\newtheorem{remark}[theorem]{Remark}					
\newtheorem{remarks}[theorem]{Remarks}
\newtheorem{proposition}[theorem]{Proposition}

\newtheorem{fact}[theorem]{Fact}

\newcommand{\abs}[1]{\left|#1\right|}		
\newcommand{\st}{\,|\,} 				 		
\newcommand{\E}{\mathop{\mathbb{E}}}  		
\newcommand{\R}{\mathbb{R}}  		
\newcommand{\N}{\mathbb{N}} 			  		
\newcommand{\F}{\mathbb{F}}					

\newcommand{\mc}{\mathcal}

\newcommand{\set}[1]{\left\{ #1 \right\}}   
\newcommand{\ip}[1]{\langle #1 \rangle}     
\newcommand{\brac}[1]{\left( #1 \right)}    
\newcommand{\sqbrac}[1]{\left[ #1 \right]}  

\newcommand{\norm}[1]{\left\vert #1 \right\vert}			
\newcommand{\Norm}[1]{\left\Vert #1 \right\Vert}			

\newcommand{\ind}{\mathds{1}}

\newcommand{\poly}{\mathrm{poly}}

\newcommand{\val}{\textnormal{val}} 	
\newcommand{\Win}{\textnormal{Win}} 	 			
\newcommand{\up}[1]{{\left(#1\right)}}				
\newcommand{\supp}{\textnormal{supp}}				
\newcommand{\ghz}{\textnormal{GHZ}}

\begin{document}

\title{Improved Parallel Repetition for GHZ-Supported Games via Spreadness}
\author{
Yang P. Liu\thanks{Department of Computer Science, Carnegie Mellon University. Email: {\tt yangl7@andrew.cmu.edu}.} \and 
Shachar Lovett\thanks{Department of Computer Science and Engineering, University of California, San Diego. Email: {\tt slovett@ucsd.edu}. Supported by Simons Investigator Award \#929894 and NSF award CCF-2425349.} \and
Kunal Mittal\thanks{Department of Computer Science, Courant Institute of Mathematical Sciences, New York University. Email: {\tt kunal.mittal@nyu.edu}. Research supported a Simons Investigator Award.}
}
\date{}
\maketitle

\begin{abstract}
	We prove that for any 3-player game $\mathcal G$, whose query distribution has the same support as the GHZ game (i.e., all $x,y,z\in \{0,1\}$ satisfying $x+y+z=0\pmod{2}$), the value of the $n$-fold parallel repetition of $\mathcal G$ decays exponentially fast: \[ \text{val}(\mathcal G^{\otimes n}) \leq \exp(-n^c)\] for all sufficiently large $n$, where $c>0$ is an absolute constant. 
	
	We also prove a concentration bound for the parallel repetition of the GHZ game: For any constant $\epsilon>0$, the probability that the players win at least a $\left(\frac{3}{4}+\epsilon\right)$ fraction of the $n$ coordinates is at most $\exp(-n^c)$, where $c=c(\epsilon)>0$ is a constant.
	
	In both settings, our work exponentially improves upon the previous best known bounds which were only polynomially small, i.e., of the order $n^{-\Omega(1)}$. Our key technical tool is the notion of \emph{algebraic spreadness} adapted from the breakthrough work of Kelley and Meka (FOCS '23) on sets free of 3-term progressions.
\end{abstract}

\newpage
\tableofcontents
\newpage

\section{Introduction}

In a $k$-player game $\mc G$, a verifier samples a tuple of questions $ (x^\up{1},\dots,x^\up{k})$ from a distribution $Q$.
Then, for each $j\in \set{1,\dots,k}$, the verifier gives the question $x^\up{j}$ to player $j$, and the player gives back an answer $a^\up{j}$, which depends only on $x^\up{j}$.
The verifier now declares whether the players win or lose based on the evaluation of a predicate $V(x,a) \in \set{0,1}$ depending on the questions $x = (x^\up{1},\dots,x^\up{k})$ and answers $a = (a^\up{1},\dots,a^\up{k})$.
We define the game value, denoted $\val(\mc G)$, as the maximum winning probability (with respect to the distribution $Q$) over all possible player strategies; see Definitions~\ref{defn:multiplayer_games},~\ref{defn:game_value} for formal definitions.

A natural question that arises is: How does the value of the game behave under parallel repetition~\cite{FRS94}?
The $n$-fold parallel repetition, denoted $\mc G^{\otimes n}$, is a game where the players play, and try to win, $n$ independent copies of the game in parallel.
More precisely, the verifier samples questions $(x_i^\up{1},\dots,x_i^\up{k})\sim Q$ independently for $i=1,\dots,n$, and for each $j\in \set{1,\dots,k}$, sends questions $(x_1^\up{j},\dots, x_n^\up{j})$ to player $j$, to which they give back answers $(a_1^\up{j},\dots, a_n^\up{j})$.
The verifier says the players win if $V((x_i^\up{1},\dots,x_i^\up{k}), (a_i^\up{1},\dots,a_i^\up{k}))=1$ for each $i\in \set{1,\dots,n}$; see Definition~\ref{defn:game_parrep} for a formal definition.

Note that for any game $\mc G$, it holds that $\val(\mc G^{\otimes n})\geq \val(\mc G)^n$, since the players can achieve this value by repeating an optimal strategy in each of the $n$ coordinates.
Although one might expect the na\"ive bound $\val(\mc G^{\otimes n})\leq \val(\mc G)^n$ to hold as well, this turns out to be false~\cite{For89, Fei91, FV02, Raz11}.
Roughly speaking, this failure occurs because the players do not have to treat the $n$ copies of the game independently, and can instead correlate their answers among different copies.
Remarkably, it turns out that this failure is intimately connected to  the geometry of high-dimensional Euclidean tilings~\cite{FKO07, KORW08, AK09, BM21}.

Parallel repetition of 2-player games is well-understood.
Raz~\cite{Raz98} showed that for any game with value less than 1, the value of the $n$-fold parallel repetition decays exponentially in $n$.
Subsequent works have simplified this proof and strengthened the quantitative bounds~\cite{Hol09, BRRRS09, Rao11, RR12, DS14, BG15}.
These and related works have led to several applications in various domains, including the theory of interactive proofs~\cite{BOGKW88}, PCPs and hardness of approximation~\cite{FGLSS96, ABSS97, ALMSS98, AS98, BGS98, Fei98, Has01, Kho02a, Kho02b, GHS02, DGKR05, DRS05}, quantum information~\cite{CHTW04,BBLV13}, and communication complexity~\cite{PRW97, BBCR13, BRWY13}.
The reader is referred to the survey~\cite{Raz10} for more details.

Parallel repetition of multiplayer games is much less understood.
The only general bound says that for any game $\mc G$ with value less than 1, it holds that $\val(\mc G^{\otimes n}) \leq 1/\alpha(n)$, where $\alpha(n)$ is a very slowly growing inverse-Ackermann function~\cite{Ver96, FK91, Pol12}.
Recent work has made substantial progress in understanding special cases of multiplayer games~\cite{DHVY17, HR20, GHMRZ21, GHMRZ22, GMRZ22, BKM23, BBKLM24, BBKLM25, BBKLMM25}; however, the general question remains wide open, even for 3-player games.

Proving improved parallel repetition for multiplayer games has several potential applications.
It is known that a strong parallel repetition theorem for a certain class of multiplayer games implies super-linear lower bounds for non-uniform Turing machines, which is a longstanding open problem in complexity theory~\cite{MR21}.
Additionally, parallel repetition in the large answer alphabet regime is equivalent to many problems in high-dimensional extremal combinatorics, such as the density Hales-Jewett problem, and the problem of square-free sets in finite fields~\cite{FV02, HHR16, Mit25}.
Also, as stated in~\cite{DHVY17}, it is believable that improved understanding of parallel repetition can lead to a better understanding of communication complexity in the number-on-forehead (NOF) model, which is intimately connected to circuit lower bounds.

The focus of this paper is the 3-player GHZ game~\cite{GHZ89}, which proceeds as follows: The verifier samples questions $(x,y,z)\in \set{0,1}^3$ uniformly at random such that $x+y+z=0\pmod 2$, and the players' goal is to give answers $a,b,c\in \set{0,1}$ respectively satisfying $a+b+c\pmod 2 = x\lor y\lor z$.
The GHZ game has played a foundational role in quantum information theory, due in part to the fact that quantum strategies can win this game with probability 1, whereas any classical strategy wins with probability at most $3/4$.
Moreover, the GHZ game satisfies a self-testing property, which says that all quantum strategies achieving value 1 are essentially the same; this has led to applications like entanglement testing and device-independent cryptography~\cite{SB20}.

The problem of parallel repetition of the GHZ game has been discussed in several recent works.
Dinur, Harsha, Venkat and Yuen~\cite{DHVY17} extend the 2-player information-theoretic techniques of Raz~\cite{Raz98} and prove parallel repetition for a specific class of multiplayer games satisfying a certain connectivity property; they identify the GHZ game as a multiplayer game which is in some sense maximally far from this class of games; regarding the difficulty of the problem, they write
\begin{quote}
	``We believe that the strong correlations present in the GHZ question distribution represent the `hardest instance' of the multiplayer parallel repetition problem.''
\end{quote}

Subsequent research established polynomial decay bounds for the $n$-fold repetition of the GHZ game, showing $\val(\textnormal{GHZ}^{\otimes n}) \leq n^{-\Omega(1)}$~\cite{HR20, GHMRZ21}.
Notably, these results apply to any 3-player game $\mc G$ whose query distribution $Q$ has the same support as the GHZ game (i.e., $\{(x,y,z)\in \set{0,1}^3:x+y+z=0\pmod{2}\}$).

More recently, an exponential decay bound was proven for the GHZ game~\cite{BKM23}, and then extended to all 3-player XOR games\footnote{In a 3-player XOR game, the answers $a,b,c$ of the three players lie in some finite Abelian group $H$, and the predicate is of the form $a+b+c=\varphi(x,y,z)$, for some function $\varphi$ mapping questions $(x,y,z)$ of the players to elements of $H$.} satisfying a certain distributional assumption~\cite{BBKLM25}.
However, these works rely heavily on the XOR structure of the game predicate, exploiting it via sophisticated tools from Fourier analysis and additive combinatorics.
For instance, in the case of the GHZ game, they observe that the predicate corresponds to a linear equation over the group $\mathbb{Z}/4\mathbb{Z}$.
Consequently, these techniques do not extend to general game predicates lacking such XOR structure. 

In this work, we improve upon the existing polynomial decay bounds~\cite{HR20, GHMRZ21} by establishing a stretched exponential bound for all games sharing the query support of the GHZ game.
Crucially, unlike the aforementioned works on XOR games~\cite{BKM23, BBKLM25}, our result is agnostic to the answer sets and the game predicate---it solely depends on the support of the query distribution.
Our technical approach is distinct as well: we utilize spreadness-based arguments from recent works~\cite{KM23, KLM24, JLLOS25}, whose application to the context of parallel repetition is novel.

Our main result is the following:
\begin{theorem}[Parallel Repetition for the GHZ query support]\label{thm:intro_main}
	Let $\mc G$ be any 3-player game with value $\val(\mc G)<1$, whose query distribution has support \[\set{(x,y,z)\in \set{0,1}^3:x+y+z=0\pmod{2}}.\]
	Then, for all sufficiently large $n$,\footnote{i.e., $n\geq N$, where $N$ is a constant depending on the game $\mc G$.\label{footnote:suff_large_n}} it holds that
	\[\val(\mc G^{\otimes n}) \leq \exp\brac{-n^{c}},\]
	where $c>0$ is an absolute constant.
\end{theorem}

We complement this result with a concentration bound for games with the GHZ query distribution.
To the best of our knowledge, the exponential decay bounds established in prior works~\cite{BKM23, BBKLM24} do not imply such concentration.
We prove the following:

\begin{theorem}[A Concentration Bound; restated and proved as Theorem~\ref{thm:ghz_conc}]\label{thm:intro_conc}
	Let $\mc G$ be any 3-player game with value $\val(\mc G)<1$, whose query distribution is uniform over the set \[\set{(x,y,z)\in \set{0,1}^3:x+y+z=0\pmod{2}}.\]
	Then, for every constant $\epsilon>0$, there exists $c=c(\epsilon)>0$, such that for all sufficiently large $n$,\footnote{i.e., $n\geq N$, where $N$ is a constant depending on the game $\mc G$ and the parameter $\epsilon$.\label{footnote:suff_large_n_eps}} the probability that the players can win at least $\val(\mc G)+\epsilon$ fraction of the $n$ coordinates in the game $\mc G^{\otimes n}$ is at most $\exp\brac{-n^{c}}$.
\end{theorem}

As a direct consequence, we obtain a concentration bound for the standard GHZ game.

\begin{corollary}
	For every constant $\epsilon>0$, there exists $c=c(\epsilon)>0$, such that the probability of the players winning at least $\brac{\frac{3}{4}+\epsilon}n$ coordinates in the game $\ghz^{\otimes n}$ is at most $\exp\brac{-n^{c}}$.
\end{corollary}

It remains an interesting open problem to establish analogous concentration bounds (as in Theorem~\ref{thm:intro_conc}) for games over the GHZ support where the underlying distribution is not uniform.\footnote{We remark that Theorem~\ref{thm:intro_conc}, combined with a reduction similar to Lemma~\ref{lemma:wlog_unif_dist} (also see Footnote~\ref{footnote:wlog_unif}), already implies the following: the probability that the players win at least $1-\beta+\epsilon$ fraction of the $n$ coordinates is at most $\exp(-n^{\Omega_{\epsilon}(1)})$, where $\beta = \beta(\mc G)>0$ is a constant.}

\subsection{Organization}

In Section~\ref{sec:overview}, we give an overview of our proofs.
In Section~\ref{sec:prelims}, we establish some preliminaries.
In Section~\ref{sec:alg_spread}, we define a notion of pseudorandomness, called \emph{algebraic spreadness}, that will be useful throughout this paper; we also show how to decompose arbitrary sets into components satisfying this spreadness condition.
In Section~\ref{sec:unif_sq_covers}, we show that any \emph{diagonal-product} set composed of algebraically spread sets is \emph{uniformly covered} by \emph{squares}.
Finally, in Section~\ref{sec:ghz_parrep}, we use the results of the previous sections and prove parallel repetition for games with the GHZ query support (Theorem~\ref{thm:intro_main}).
In Section~\ref{sec:ghz_conc}, we prove our concentration bound (Theorem~\ref{thm:intro_conc}).

\section{Overview}\label{sec:overview}

In this section, we outline the proof of our main result (Theorem~\ref{thm:intro_main}).

\subsection{General Inductive Framework and High-Level Approach}

We begin by introducing a general inductive framework for parallel repetition.
Our proof follows the basic setup established by Raz for parallel repetition of 2-player games~\cite{Raz98}.

Let $\mc G$ be a 3-player game, whose query distribution $Q$ has support \[\supp(Q)=\set{(x,y,z)\in \F_2^3 : x+y+z=0}.\footnote{The framework applies to all $k$-player games, but we restrict to 3 player games with the GHZ queries.}\]
Here, we identified $\set{0,1}$ with the finite field $\F_2$.
For the purpose of parallel repetition, we may assume without loss of generality that $Q$ is the uniform distribution over its support (see Lemma~\ref{lemma:wlog_unif_dist}).\footnote{\label{footnote:wlog_unif}Roughly speaking, this is because any distribution $Q$ \emph{contains} a small copy of the uniform distribution over $\supp(Q)$, and hence winning $n$ copies under the distribution $Q$ is at least as hard as winning $\Omega(n)$ copies under the the uniform distribution.} 
Now, consider the game $\mc G^{\otimes n}$, and consider any strategy for the 3 players in this game.
For each $i\in [n]$, let $\Win_i$ be the event that the players win the $i$\textsuperscript{th} coordinate.
By the chain rule, for any permutation $i_1,i_2,\dots,i_n$ of $[n]$, we can write the winning probability as
\[ \Pr[\Win_1\land\Win_2\land\dots\land \Win_n] = \prod_{k=1}^n \Pr[\Win_{i_k}\mid \Win_{i_1}\land\dots\land\Win_{i_{k-1}}]. \]
To bound this, we proceed inductively.
Assuming the players have won a set of coordinates $i_1,\dots,i_k$, we aim to identify a \emph{hard coordinate} $i\in [n]$ whose conditional winning probability is at most $1-\Omega(1)$.  
Formally, we wish to prove the following condition:\footnote{Conditioning on the event $\Win_{i_1}\land\dots \land \Win_{i_k}$ can introduce complex correlations among the players' inputs. However, this event depends deterministically on the questions and answers of the players in coordinates $i_1,\dots,i_k$, and hence instead of conditioning on the event $\Win_{i_1}\land\dots \land \Win_{i_k}$, we can condition on typical questions and answers to the players in coordinates $i_1,\dots,i_k$. This induces a product event on the players' input.}

\begin{description}
	\item[Inductive Step:] For every product event $\mc E = E\times F\times G \subseteq\F_2^n\times \F_2^n\times \F_2^n$, with measure $\Pr[\mc E]\geq \alpha$, there exists a coordinate $i\in [n]$ such that $\Pr[\Win_i\mid \mc E] \leq 1-\Omega(1)$.
\end{description}
Establishing this condition for $\alpha = \exp(-n^{\Omega(1)})$ implies the desired bound on $\val(\mc G^{\otimes n})$  via the inductive strategy above (see Lemma~\ref{lemma:inductive_criteria}).

For the remainder of this overview, consider any product event $\mc E = E\times F\times G \subseteq (\F_2^n)^3$, with measure $\Pr[\mc E]\geq \alpha = \exp(-n^{\Omega(1)})$.
Our goal is to identify a coordinate $i\in [n]$ that remains hard to win when the inputs are conditioned on $\mc E$.
The proof proceeds in two main steps:
\begin{enumerate}
	\item \textbf{Identify Hard Sets:} We define a class of sets called \emph{squares}, such that if the input distribution is restricted to a square, many coordinates are hard to win.
	\item \textbf{Distributional Approximation:} We show that the conditional distribution $Q^{\otimes n}|\mc E$ can be approximated (in $\ell_1$ distance) by a convex combination of such square distributions. 
\end{enumerate}
Combining these steps establishes that $\Pr[\Win_i\mid \mc E] \leq 1-\Omega(1)$ for randomly chosen $i\in [n]$.

\subsection{Step 1: Squares are Hard}\label{sec:hard_set_sq}

Since the third player's input in $\mc G^{\otimes n}$ is fully determined by the inputs to the first two players (inputs $(x,y,z)\in \supp(Q)^n\subseteq (\F_2^n)^3$ satisfy $z=x+y$), we can analyze the game primarily by only looking at inputs of the first two players.
The restriction of a product event $\mc E= E\times F\times G$ to the first two players is captured by a \emph{diagonal-product set}, defined as follows:

\begin{definition}[Diagonal-Product Set]\label{defn:diag_prod}
	Given sets $X,Y,Z\subseteq \F_2^n$, we define the corresponding diagonal-product set, denoted $S(X,Y,Z)$, as \[S(X,Y,Z) = \set{(x,y)\in \mc \F_2^n\times \F_2^n : x\in X, y\in Y, x+y\in Z}.\]
\end{definition}

Sampling inputs $(x,y,z)\sim  Q^{\otimes n}|\mc E$ is the same as sampling $(x,y)\sim S(E,F,G)$ uniformly and setting $z=x+y$.

We now define our hard sets, that are \emph{squares} in $\F_2^n\times \F_2^n$:

\begin{definition}[Square]\label{defn:square}
	A \emph{square} $s_{x,y,w}\subseteq \F_2^n\times \F_2^n$, for $x,y,w\in \F_2^n$, is the set \[s_{x,y,w} = \set{(x,y), (x+w,y), (x,y+w), (x+w,y+w)}.\]
\end{definition}
\begin{remark}\label{remark:sq_rep}
	Given a square $s\subseteq \F_2^n\times \F_2^n$, suppose we wish to represent it as $s=s_{x,y,w}$.
	Note that the width $w$ is uniquely determined by the square $s$.
	If $w=0$, the square contains a single point, in which case $(x,y)$ is uniquely determined.
	If $w\not=0$, the square $s_{x,y,w}$ has 4 different representations, given by $s_{x,y,w}=s_{x+w,y,w}=s_{x,y+w,w}=s_{x+w,y+w,w}$.
\end{remark}

The crucial property of a square is its local hardness:
Consider a square $s = s_{x_0,y_0,w}$, and a coordinate $i\in [n]$ such that $w_i\not=0$ (a \emph{non-trivial} coordinate).\footnote{A coordinate is non-trivial if and only if the points in $s$ do not have a constant value in this coordinate.}
Suppose the inputs $(x,y,z)$ to the three players are sampled from the square $s$ as follows: let $(x,y)\sim s$ be chosen uniformly at random, and let $z=x+y$.
Under this distribution, no strategy of the players can win coordinate $i$ with probability more than $\val(\mc G)$.
Roughly speaking, this is true since this distribution ``looks exactly the same'' as the base distribution $Q$, and hence the players can embed a single copy of the game $\mc G$ into coordinate $i$ of this distribution; see Lemma~\ref{lemma:non_triv_in_extens_hard} for formal details of this step.\footnote{We remark that the set of inputs $\set{(x,y,x+y): (x,y)\in s}$ was introduced in the work~\cite{GHMRZ21}, where they call this set a \emph{bow-tie}. We also note that squares correspond to the maximal \emph{forbidden-subgraphs} inside the GHZ query distribution (in the sense that no player strategy can win on all four points of a square), and the largest size of square-free sets in $\F_2^n\times \F_2^n$ exactly captures the value of the $n$-fold parallel repetition of the GHZ query distribution in the large answer alphabet regime~\cite{Mit25}.
}

\subsection{Step 2: Approximating the Distribution}\label{sec:dist_approx}

The second, and more technically demanding, step of the proof is to show that the distribution $Q^{\otimes n}|\mc E$ is well-approximated by a mixture of square distributions as above.
Equivalently, we wish to approximate the uniform distribution over the diagonal-product set $S=S(E,F,G)$, denoted $U_S$, by a convex combination of uniform distributions over squares.
Specifically, we analyze the distribution $\mu$ generated as follows: sample a square $s\subseteq S$ uniformly at random, and output an element uniformly at random from $s$.

Note that a priori it is not even clear that the set $S$ contains squares, and that the distribution $\mu$ is well-defined.
However, we utilize the concept of \emph{algebraic spreadness}, a pseudorandomness notion introduced in the breakthrough work of Kelley and Meka on bounds for sets without 3-term arithmetic progressions~\cite{KM23}.
They show that algebraically spread subsets behave like random subsets in a certain sense (see Definition~\ref{defn:alg_spread} and Theorem~\ref{thm:algspread_conv}); in our setting, we show that if all $E,F,G$ are algebraically spread, the set $S$ contains the expected density of squares.
Our approximation argument proceeds in two stages:

\paragraph{Uniformization:} Given arbitrary sets $E,F,G$, we decompose the diagonal-product set $S(E,F,G)$ into disjoint components $S(E_1,F_1,G_1) \cup S(E_2,F_2,G_2) \dots \cup S(E_T,F_T,G_T)$ (plus a negligible remainder), such that within each component, the sets $E_i,F_i,G_i$ are all algebraically spread.
This generalizes the recent decomposition technique of~\cite{JLLOS25}, which achieved algebraic spreadness for only two out of the three sets and used it to prove bounds on corner-free sets in finite fields.
We extend this to all three sets via a careful recursive argument; see Section~\ref{sec:uniformization} for more details.

\paragraph{Counting and Approximation:}
Assuming that all of $E,F,G$ are algebraically spread, we demonstrate that they behave like random sets of the same density, and that the square-sampling distribution $\mu$ approximates the uniform distribution $U_S$ in $\ell_1$ distance.
We prove this by establishing tight bounds on the $\ell_2$ norm (collision probability) of the distribution $\mu$.
A key technical tool here is a recent graph-counting result of~\cite{FHHK24}, based on a pseudorandomness notion called combinatorial spreadness, which we show holds in our setting.
See Lemma~\ref{lemma:spread_imp_l1_l2} and Proposition~\ref{prop:sq_cover} for more details on this step.

\paragraph{Why Algebraic Spreadness?} A crucial difference between our work and previous polynomial decay bounds~\cite{GHMRZ21} lies in the choice of the pseudorandomness property for the sets $E,F,G$.
Prior works use the concept of Fourier-uniformity (pseudorandomness against linear tests), and yield bounds when the density $\alpha = \Pr[\mc E]$ is at least polynomial, i.e., of the order $n^{-O(1)}$.
The use of algebraic spreadness enables us to surpass this barrier and get bounds even when the density $\alpha = \exp(-n^{\Omega(1)})$ is exponentially smaller.
However, this improvement incurs a significant technical cost. Both the decomposition and approximation steps are standard under Fourier uniformity, whereas establishing them under algebraic spreadness constitutes the main technical contribution of this work.

\section{Preliminaries}\label{sec:prelims}

Let $\N = \set{1,2,\dots}$ denote the set of natural numbers. For $n\in \N$, we use $[n]$ to denote the set $\set{1,2,\dots,n}$.


\subsection{Probability Distributions}

Let $P$ be a distribution (over an underlying finite set $\Omega$, which is usually clear from context). 
We use $\supp(P) = \{\omega \in \Omega : P[\omega] > 0\}$ to denote the support of the distribution $P$.
For an event $E\subseteq \Omega$ with $P[E]>0$, we use $P|E$ to denote the conditional probability distribution $P$ conditioned on $E$.

For distributions $P$ and $Q$ over a set $\Omega$, the $\ell_1$-distance between them is defined as
\[\Norm{P-Q}_1 = \sum_{\omega\in \Omega}\norm{P[\omega]-Q[\omega]}.\]

We state a useful tail bound on the sum of independent random variables:

\begin{fact}[Chernoff Bounds, see \cite{MU05} for reference]\label{fact:chernoff}
	Let $X_1,\dots, X_n \in \set{0,1}$ be independent random variables each with mean $\mu$, and let $X = \sum_{i=1}^n X_i$. Then, for all $\delta \in (0,1)$,
	\[\Pr\sqbrac{X \leq (1-\delta)\mu n} \leq e^{-\frac{\delta^2 \mu n}{2}},\]
	\[\Pr\sqbrac{X \geq (1+\delta)\mu n} \leq e^{-\frac{\delta^2 \mu n}{3}}.\]
\end{fact}

We also state another useful lemma:

\begin{lemma}\label{lemma:large_subset_random}
	Let $t,n\in \N,\ t< n$.
	Let $\mu$ be the uniform distribution on $[n]$, and let $\nu$ be the uniform distribution on $[n-t]$.
	Then, $\Norm{\mu-\nu}_1 = 2t/n$. 
\end{lemma}
\begin{proof}
	We have $\Norm{\mu-\nu}_1 = \sum_{i=1}^{n-t}\brac{\frac{1}{n-t}-\frac{1}{n}}+ \sum_{i=n-t+1}^n \frac{1}{n} = \frac{2t}{n}$.
\end{proof}


\subsection{Vector Spaces over GF(2)}

Let $\mc V$ be a finite dimensional vector space over $\F_2$, with the uniform measure; often we shall have $\mc V=\F_2^n$ for some $n\in \N$.

\begin{definition}[Inner Product]
	For functions $f,g:\mc V\to \R$, define their inner product as \[ \ip{f,g} = \E_{x\sim \mc V}[f(x)g(x)] ,\]
	where $x\sim \mc V$ denotes that $x$ is uniformly chosen from $\mc V$.
\end{definition}

\begin{definition}[$L^p$ norm]
	For $f:\mc V\to\R$, and $p\geq 1$, we define
	\[ \Norm{f}_p = \E[f^p]^{1/p}. \]
\end{definition}

\begin{definition}[Convolution]
	For $f,g:\mc V\to\R$, we define their convolution as
	\[ (f*g)(x) = \E_{y\sim \mc V}[f(y)g(x+y)]. \]
\end{definition}

We observe the following simple fact:
\begin{fact}
	For functions $f,g,h:\mc V\to \R$,
	\[ \ip{f*g, h} = \ip{f*h,g} = \ip{g*h, f} = \E_{x,y\sim \mc V}[f(x)g(y)h(x+y)].\]
\end{fact}

\begin{definition}[Density Function]
	For a nonempty set $A\subseteq \mc V$, we define its density function as
	\[\varphi_A = \frac{\ind_A}{\E\sqbrac{\ind_A}} = \frac{\ind_A}{\abs{A}/\abs{\mc V}} .\]
    Note that $\E[\varphi_A]=1$.
\end{definition}


\subsection{3-Player Games and Parallel Repetition}

We overview some basic definitions regarding multiplayer games.
We shall restrict our focus to 3-player games.

\begin{definition}[3-Player Game]\label{defn:multiplayer_games} 
	A $3$-player game $\mc G$ is a tuple $\mc G = (\mc X\times \mc Y\times \mc Z,\ \mc A\times \mc B\times \mc C,\ Q, V_{pred})$, where the question sets $\mc X,\mc Y, \mc Z$ and the answer sets $\mc A, \mc B, \mc C$ are finite sets, $Q$ is a probability distribution over $\mc X\times \mc Y\times \mc Z$, and $V_{pred}:(\mc X\times \mc Y\times \mc Z)\times (\mc A\times \mc B\times \mc C)\to\set{0,1}$ is a predicate.
\end{definition}

The game $\mc G$ proceeds as follows: A verifier samples questions $(X,Y,Z)\sim Q$; then, the verifier sends $X$ to player 1, $Y$ to player 2, and $Z$ to player 3, to which the players respond back with answers $A\in \mc A,\ B\in \mc B,\ C\in \mc C$ respectively.
Finally, the verifier declares that the players win if and only if $V_{pred}( (X,Y,Z),(A,B,C)) = 1$.

\begin{definition}[Game Value]\label{defn:game_value}
	Let $\mc G = (\mc X\times \mc Y\times \mc Z,\ \mc A\times \mc B\times \mc C,\ Q, V_{pred})$ be a 3-player game.
	The value of the game $\mc G$, denoted $\val(\mc G)$, is defined as
	\[ \val(\mc G) = \max_{f,g,h}\Pr_{(X,Y,Z)\sim Q}\sqbrac{V_{pred}\brac{(X,Y,Z),(f(X),g(Y),h(Z))}=1}, \] 
	where the maximum is over \emph{player strategies} $f:\mc X\to \mc A,\ g:\mc Y\to \mc B,\ h:\mc Z\to \mc C$.    
\end{definition}

\begin{fact}
	The game value is unchanged even if the players are allowed to use public and private randomness, since there always exists some optimal fixed values of the random strings.
\end{fact}

Next, we define the parallel repetition of a $3$-player game, which corresponds to playing $n$ independent copies of the game in parallel.

\begin{definition}[Parallel Repetition]\label{defn:game_parrep} 
	Let $\mc G = (\mc X\times \mc Y\times \mc Z,\ \mc A\times \mc B\times \mc C,\ Q, V_{pred})$ be a 3-player game. We define its $n$-fold repetition as $\mc G^{\otimes n} = (\mc X^n\times \mc Y^n\times \mc Z^n,\ \mc A^n\times \mc B^n\times \mc C^n,\ Q^{\otimes n}, V_{pred}^{\otimes n})$.
	The distribution $Q^{\otimes n}$ is the $n$-fold product of the distribution $Q$ with itself, i.e., $Q^{\otimes n}[(x,y,z)] = \prod_{i=1}^n Q[(x_i,y_i,z_i)]$ for each $x\in \mc X^n, y\in \mc Y^n, z\in \mc Z^n$.
	The predicate $V_{pred}^{\otimes n}$ is defined as $V_{pred}^{\otimes n}((x,y,z), (a,b,c)) = \bigwedge_{i=1}^n V_{pred}((x_i,y_i,z_i), (a_i,b_i,c_i))$.
\end{definition}

\subsubsection{Some Basic Results on Parallel Repetition}

We state a lemma from \cite{FV02}, which shows that it suffices to prove parallel repetition in the case when the game's distribution is uniform over its support:

\begin{lemma}[{\cite[Lemma 3.14]{GHMRZ22}}]\label{lemma:wlog_unif_dist}
	Let $\mc G = (\mc X\times \mc Y\times \mc Z,\ \mc A\times \mc B\times \mc C,\ Q, V_{pred})$ be a $3$-player game such that $\val(\mc G)<1$.
	Let $\tilde{\mc G} = (\mc X\times \mc Y\times \mc Z,\ \mc A\times \mc B\times \mc C,\ U, V_{pred})$,  where $U$ is the uniform distribution over $\supp(Q)$.
	Then, $\val(\tilde{\mc G}) < 1$, and there exists a constant $c = c(\mc G)>0$, such that for every sufficiently large $n\in \N$,
	\[ \val(\mc G^{\otimes n}) \leq 2\cdot \val(\tilde{\mc G}^{\otimes \lfloor c n\rfloor}). \]
\end{lemma}

We state an inductive parallel repetition criterion from \cite{Raz98}:

\begin{lemma}[{\cite[Lemma B.1]{BBKLMM25}}]\label{lemma:inductive_criteria}
	Let $\mc G = (\mc X\times \mc Y\times \mc Z,\ \mc A\times \mc B\times \mc C,\ Q, V_{pred})$ be a $3$-player game, and consider its $n$-fold repetition $\mc G^{\otimes n} = (\mc X^n\times \mc Y^n\times \mc Z^n,\ \mc A^n\times \mc B^n\times \mc C^n,\ Q^{\otimes n}, V_{pred}^{\otimes n})$ for some sufficiently large $n\in \N$.
	Fix optimal strategies for the $3$ players in this game, and for each $i\in [n]$, let $\Win_i$ be the event that this strategy wins the $i$\textsuperscript{th} coordinate of the game.
	
	Let $\epsilon>0$ be a constant, and $\alpha\in (0,1],\ \alpha\geq 2^{-n}$ be such that the following condition holds: For every product event $E\times F\times G \subseteq \mc X^n\times \mc Y^n\times \mc Z^n$ with $\Pr_{Q^{\otimes n}}[E\times F\times G]\geq \alpha$, there exists a coordinate $i\in [n]$ such that $\Pr\sqbrac{\Win_i\st E\times F\times G} \leq 1-\epsilon$.
	
	Then, for some constant $c=c(\mc G)$, it holds that $\val(\mc G^{\otimes n}) \leq \alpha^c$.
\end{lemma}


\section{Algebraic Spreadness and Uniformization}\label{sec:alg_spread}

In this section, we overview a notion of pseudorandomness called algebraic spreadness.
We also show how to decompose arbitrary sets into components that satisfy this spreadness definition.


\subsection{Algebraic Spreadness}

\begin{definition}[Algebraic Spreadness]\label{defn:alg_spread}
Let $\mc V\subseteq \F_2^n$ be an affine subspace.
We say that a subset $A\subseteq \mc V$ is $(r,\epsilon)$-algebraically spread within $\mc V$ if for all affine subspaces $\mc V'\subseteq \mc V$ satisfying $\dim(\mc V')\geq \dim(\mc V)-r$, it holds that
\[ \frac{\abs{A\cap \mc V'}}{\abs{\mc V'}} \leq (1+\epsilon)\cdot \frac{\abs{A}}{\abs{\mc V}}. \] 
\end{definition}

We state a useful result about spread subsets.

\begin{theorem}\label{thm:algspread_conv}
	Let $d\geq 1, \epsilon \in (0,1/4)$.
	Then, there exists a sufficiently large integer $r = d^8 \epsilon^{-O(1)}$, and a sufficiently small $\delta = \Omega(\epsilon)$, such that the following holds:
	Suppose $A,B,C\subseteq \F_2^n$ are sets each of size at least $2^{-d}\cdot \abs{\F_2^n}$. Then,
	
	\begin{enumerate}[label=(\roman*)]
		\item\label{item:km_two_spread} (\cite[Proposition 2.16]{KM23}) If at least two of  $A,B,C$ are $(r,\delta)$-algebraically spread, 
			\[ \abs{\ip{\varphi_A* \varphi_B,\ \varphi_C}-1} \leq \epsilon.\footnote{For example, if both $A,B$ are algebraically spread, we have $\abs{\ip{\varphi_A* \varphi_B,\ \varphi_C}-1} \leq 2 \Norm{\varphi_A* \varphi_B-1}_{d} \leq \epsilon$, where the first inequality follows from H\"older's inequality, and the second inequality follows from \cite[Proposition 2.16]{KM23}.}\]
		\item\label{item:km_one_spread} (\cite[Proposition 4.10]{KM23}) If at least one of $A,B,C$ is $(r,\delta)$-algebraically spread,
			\[ \ip{\varphi_A* \varphi_B,\ \varphi_C} \leq 1+\epsilon.  \]
	\end{enumerate}
\end{theorem}


\subsection{Uniformization}\label{sec:uniformization}

In this section, we describe a procedure that takes an arbitrary set $X \subseteq \F_2^n$ and approximately decomposes it into components that are algebraically spread. 
More generally, we find a simultaneous decomposition for three sets $X,Y,Z$, compatible with the diagonal-product:

\begin{proposition}\label{prop:uniformization}
	Let $r\in \N$, $\epsilon,\eta\in (0,1/10)$, and let $X,Y,Z\subseteq \mc V$ for a linear subspace $\mc V\subseteq\F_2^n$.
	Define $\alpha=\abs{S(X,Y,Z)}/{\abs{\mc V}^2}$.
	Then, there exists $C_{\eta} = (1/\eta)^{O(1)}\geq 1$ (that depends only on $\eta$), an integer $T\in \N$, and for each $i\in [T]$, a linear subspace $\mc V_i\subseteq \mc V$, points $x_i,y_i\in \mc V/\mc V_i$, and subsets $X_i\subseteq x_i+\mc V_i,\ Y_i\subseteq y_i+\mc V_i,\ Z_i\subseteq x_i+y_i+\mc V_i$, such that:
	\begin{enumerate}
		\item $\dim(\mc V_i) \geq \dim(\mc V)- r\epsilon^{-3}\log_2(4/\alpha)^{C_{\eta}}$ for all $i\in [T]$.
		\item $S(X_1,Y_1,Z_1),\dots,S(X_T,Y_T,Z_T)$ are disjoint subsets of $S(X,Y,Z)$ such that 
			\[ \abs{ S(X,Y,Z)\setminus \cup_{i=1}^T S(X_i,Y_i,Z_i)} \leq  \eta\abs{S(X,Y,Z)}. \]
		The diagonal-product sets are defined as in Definition~\ref{defn:diag_prod}.
		\item For all $i\in [T]$, we have $\abs{X_i},\abs{Y_i},\abs{Z_i} \geq 2^{-(\log_2(4/\alpha))^{C_{\eta}}}\cdot \abs{\mc V_i}$, and $X_i,Y_i,Z_i$ are $(r,\epsilon)$-algebraically spread within $x_i+\mc V_i,y_i+\mc V_i,x_i+y_i+\mc V_i$ respectively.
	\end{enumerate}
\end{proposition}

The proof of the above proposition is based on~\cite[Section 5.3]{JLLOS25}, with the exception that we work to make all the 3 sets spread (instead of just 2).
As a result, we obtain a weaker dependence on the parameter $\eta$ compared to their two set version; however, this is inconsequential in our setting, as we use this only for constant $\eta$.
In the following subsections, we first show how to decompose one set, then two sets, and finally three sets.

\subsubsection{Uniformization for 1 Set}

We first show how to decompose a single set $X$ into spread components.
Before that, we show that any large set contains a spread set in it.

\begin{lemma}[Existence of a Spread Subset]\label{lemma:spread_inside_set}
	Let $r\in \N,\ \epsilon\in (0,1)$, and let $X\subseteq \mc V$ for a linear subspace $\mc V\subseteq\F_2^n$.
	Define $\alpha_X=\abs{X}/\abs{\mc V}$.
	Then, there exists an affine subspace $\mc V'\subseteq \mc V$ such that:
	\begin{enumerate}
		\item $\dim(\mc V') \geq \dim(\mc V)- r\epsilon^{-1}\log_2(1/\alpha_X).$
		\item The set $X' = X\cap\mc V'$ satisfies $\frac{\abs{X'}}{\abs{\mc V'}}\geq \alpha_X$, and $X'$ is $(r,\epsilon)$-algebraically spread within $\mc V'$.
	\end{enumerate}
\end{lemma}
\begin{proof}
	We proceed iteratively. Let $\mc V^{(0)} = \mc V$; for $t=0,1,2,\dots$ we do the following: if $X\cap \mc V^{(t)}$ is $(r,\epsilon)$-algebraically spread within $\mc V^{(t)}$, we stop; otherwise, there exists some subspace $\mc V^{(t+1)}\subseteq \mc V^{(t)}$ such that $\dim(\mc V^{(t+1)})\geq \dim(\mc V^{(t)})-r$, and 
	\[ \frac{\abs{X\cap \mc V^{(t+1)}}}{\abs{\mc V^{(t+1)}}} > (1+\epsilon)\cdot \frac{\abs{X\cap \mc V^{(t)}}}{\abs{\mc V^{(t)}}}.\] 
	If we have not stopped till step $t$, we have $\dim(\mc V^{(t)}) \geq \dim(\mc V)-rt$, and \[ \frac{\abs{X\cap \mc V^{(t+1)}}}{\abs{\mc V^{(t+1)}}} \geq (1+\epsilon)^t \alpha_X \geq 2^{\epsilon t}\alpha_X. \]
	Since this density cannot exceed 1, we must stop for some $t\leq \epsilon^{-1}\log_2(1/\alpha_X)$. 
\end{proof}

Using the above lemma, we show how to decompose a single set into spread components.

\begin{lemma}[Uniformization for 1 Set]\label{lemma:uniformization_one_set}
	Let $r\in \N,\ \epsilon,\eta\in (0,1)$.
	Let $X\subseteq \mc V$ for a linear subspace $\mc V\subseteq\F_2^n$.
	Then, there exists $T\in \N$, and for each $i\in [T]$, an affine subspace $\mc V_i\subseteq \mc V$ and a subset $X_i\subseteq \mc V_i$, such that:
	\begin{enumerate}
		\item $\dim(\mc V_i) \geq \dim(\mc V)- r\epsilon^{-1}\log_2(1/\eta)$ for all $i\in [T]$.
		\item $X_1,\dots,X_T$ are disjoint subsets of $X$ such that $\abs{X\setminus \cup_{i\in [T]}X_i} \leq \eta \abs{\mc V}$.
		\item For all $i\in [T]$, $\abs{X_i} \geq \eta\abs{\mc V_i}$ and $X_i$ is $(r,\epsilon)$-algebraically spread within $\mc V_i$.
	\end{enumerate}
\end{lemma}
\begin{proof}
	We proceed iteratively, via the following algorithm:
	Let $X^{(1)} = X,\ \mc V^{(1)}=\mc V$.
	At any time step $t$, if $\abs{X^{(t)}} \leq \eta\abs{\mc V}$, we stop.
	Else, if $\abs{X^{(t)}} \geq \eta\abs{\mc V}$, by Lemma~\ref{lemma:spread_inside_set}, we find an affine subspace $\mc V_t\subseteq \mc V$ of dimension $\dim(\mc V_t)\geq \dim(\mc V)-r\epsilon^{-1}\log_2(1/\eta)$, such that $X_t:=X^{(t)}\cap \mc V_t$ satisfies $\abs{X_t}\geq \eta\abs{\mc V_t}$, and $X_t$ is $(r,\epsilon)$-algebraically spread within $\mc V_t$.
	Then, we define $X^{(t+1)} = X^{(t)}\setminus X_t$.
	
	Finally, suppose we stop at step $T+1$.
	Then, we can write $X$ as a disjoint union $X= X^{(T+1)}\cup X_1\cup X_2\cup\dots\cup X_T$, with $\abs{X^{(T+1)}}\leq \eta\abs{\mc V}$.
\end{proof}
\begin{remarks}
	Note that	
	\begin{enumerate}
		\item It must hold that $T\leq 2^{(1+r/\epsilon)\log_2(1/\eta)}$, since for all $i\in [T]$, \[ \frac{\abs{ X_i}}{\abs{\mc V}} = \frac{\abs{ X_i}}{\abs{\mc V_i}} \cdot \frac{\abs{\mc V_i}}{\abs{\mc V}} \geq \eta  \cdot 2^{- r\log_2(1/\eta)/\epsilon} = 2^{-(1+r/\epsilon)\log_2(1/\eta)}, \]
		and $\sum_{i=1}^T \abs{X_i} \leq \abs{X} \leq \abs{\mc V}$.
		\item It is useful to think of $\eta = \eta' \cdot \frac{\abs{X}}{\abs{\mc V}}$, to get a multiplicative approximation for the set $X$.
	\end{enumerate}
\end{remarks}

\subsubsection{Uniformization for 2 Sets}

Next, we state a lemma about decomposing the product $X\times Y$ of two sets $X,Y$ into spread components.

\begin{lemma}[Uniformization for 2 Sets, {\cite[Lemma 5.9]{JLLOS25}}]\label{lemma:uniformization_two_set}
	Let $r\in \N$, $\epsilon, \eta\in (0,1/10)$, and let $X,Y\subseteq \mc V$ for a linear subspace $\mc V\subseteq \F_2^n$.
	Define $\alpha_{XY} = \abs{X}\abs{Y}/{\abs{\mc V}}^2$.
	Then, there exists $T\in \N$, and for each $i\in [T]$, a linear subspace $\mc V_i\subseteq \mc V$, points $x_i,y_i\in \mc V/\mc V_i$, and subsets $X_i\subseteq x_i+\mc V_i,\ Y_i\subseteq y_i+V_i$, such that:
	\begin{enumerate}
		\item $\dim(\mc V_i) \geq \dim(\mc V)- O\brac{r\epsilon^{-2}\log_2(1/\alpha_{XY})^2\log_2(1/\eta)+r\epsilon^{-2}\log_2(1/\eta)^5}$ for all $i\in [T]$.
		\item $X_1\times Y_1,\dots,X_T\times Y_T$ are disjoint subsets of $X\times Y$  such that \[\abs{(X\times Y)\setminus \cup_{i\in [T]}(X_i\times Y_i)} \leq \eta \abs{X\times Y}\]
		\item For all $i\in [T]$, we have $\abs{X_i\times Y_i} \geq \kappa \cdot  \abs{\mc V_i}^2$ for $\kappa \geq 2^{-O(\log_2(1/\eta)^2)}\cdot \alpha_{XY}$, and $X_i,Y_i$ are $(r,\epsilon)$-algebraically spread within $x_i+\mc V_i, y_i+\mc V_i$ respectively.
	\end{enumerate}
\end{lemma}


\subsubsection{Uniformization for 3 Sets}

We shall prove Proposition~\ref{prop:uniformization} and show how to decompose three sets into spread components, compatible with the diagonal-product.
We start by stating a simple fact about the size of a  diagonal-product $S(X,Y,Z)$, which along with Theorem~\ref{thm:algspread_conv} allows us to control this size when at least one or two of $X,Y,Z$ are spread.

\begin{fact}\label{fact:diag_prod_size}
	Let $X,Y,Z\subseteq \mc V$ for a linear subspace $\mc V\subseteq \F_2^n$, and let $\alpha_X = \abs{X}/{\abs{\mc V}}$, $\alpha_Y = \abs{Y}/{\abs{\mc V}}$, $\alpha_Z = \abs{Z}/{\abs{\mc V}}$.
	Then, $\abs{S(X,Y,Z)}  = \alpha_X\alpha_Y\alpha_Z \cdot \ip{\varphi_X* \varphi_Y,\ \varphi_Z}.$
\end{fact}

First, we show a one round partitioning result for $S(X,Y,Z)$:

\begin{lemma}\label{lemma:one_round_part}
	Let $\eta\in (0,1/50)$, and let $X,Y,Z\subseteq \mc V$ for a linear subspace $\mc V\subseteq \F_2^n$.
	Define $\alpha = \abs{S(X,Y,Z)}/\abs{\mc V}^2$.
	Then, there exists an integer $r' = O(\log_2(1/(\eta\alpha))^{16})$, and constant $\epsilon' = \Omega(1)$, such that the following hold:
	
	Let $r\in \N$, $\epsilon\in (0,1/10)$ satisfy $r\geq r',\ \epsilon\leq\epsilon'$.
	Then, there exists $T\in \N$	, a set $\mc G\subseteq [T]$, and for each $i\in [T]$, a linear subspace $\mc V_i\subseteq \mc V$, points $x_i,y_i\in \mc V/\mc V_i$, and subsets $X_i\subseteq x_i+\mc V_i, Y_i\subseteq y_i+\mc V_i, Z_i\subseteq x_i+y_i+\mc V_i$, such that
	\begin{enumerate}
		\item $\dim(\mc V_i)\geq \dim(\mc V) - O\brac{r\epsilon^{-3}\log_2(1/(\eta\alpha))^6}$ for all $i\in [T]$.
		\item $S(X_1,Y_1,Z_1),\dots,S(X_T,Y_T,Z_T)$ are disjoint subsets of $S(X,Y,Z)$ such that 
			\[ \abs{ S(X,Y,Z)\setminus \cup_{i=1}^T S(X_i,Y_i,Z_i)} \leq 4 \eta\abs{S(X,Y,Z)}. \]
		\item $\abs{S(X_i,Y_i,Z_i)} \geq 2^{-O(\log_2(1/(\eta\alpha))^2)}\cdot \abs{\mc V_i}^2$ for all $i\in [T]$.
		\item For all $i\in \mc G$, the sets $X_i,Y_i,Z_i$ are $(r,\epsilon)$-algebraically spread within $x_i+\mc V_i,y_i+\mc V_i,x_i+y_i+\mc V_i$ respectively.
		\item $\sum_{i\in \mc G} \abs{S(X_i,Y_i,Z_i)} \geq \frac{1}{10}\cdot \abs{S(X,Y,Z)}$.
	\end{enumerate}
\end{lemma}
\begin{proof}
	Let $\kappa= 2^{-O(\log_2(1/(\eta\alpha))^2)}\cdot \alpha_{XY} \geq 2^{-O(\log_2(1/(\eta\alpha))^2)}$ be as in the third item in Lemma~\ref{lemma:uniformization_two_set}, when applied with parameter $\eta\alpha$; we may assume $\kappa \leq \eta\alpha/2$.
	Let $r' = O(\log_2(1/\kappa)^8),\ \epsilon'=\Omega(1)$ be as in Theorem~\ref{thm:algspread_conv} with density $\kappa/2$ and error $1/10$ (that is, with $d\gets \log_2(2/\kappa)$ and $\epsilon\gets 1/10$).

	Let $r_0 = r + \lceil r\epsilon^{-1}\log_2(1/(\eta\alpha)) \rceil$.
	By Lemma~\ref{lemma:uniformization_two_set} applied to $X\times Y$, with parameters $r_0,\epsilon/10, \eta\alpha$, we may find $T\in \N$, and for each $t\in [T]$ a linear subspace $\mc V_t$ of dimension $ \dim(\mc V_t)\geq\dim(\mc V)- O\brac{r_0\epsilon^{-2}\log_2(1/(\eta\alpha))^5}$, points $x_t,y_t\in \mc V/\mc V_t$, and subsets $X_t\subseteq x_t+\mc V_t, Y_t\subseteq y_t+\mc V_t$ of size $\abs{X_t}\abs{Y_t}/\abs{\mc V_t}^2 \geq \kappa$, such that $X_t,Y_t$ are $(r_0,\epsilon/10)$-algebraically spread inside $x_t+\mc V_t,\ y_t+\mc V_t$ respectively.	
	Also, we have that $X_1\times Y_1,\dots,X_T\times Y_T$ are disjoint subsets of $X\times Y$, with $\widetilde{XY} = (X\times Y)\setminus \cup_{t\in [T]}(X_t\times Y_t)$ satisfying $|\widetilde{XY}| \leq \eta\alpha \abs{X\times Y}$.
	In particular, it also holds that $S(X_1,Y_1,Z),\dots,S(X_T,Y_T,Z)$ are disjoint subsets of $S(X,Y,Z)$.
	
	Now, consider any $t\in [T]$.
	Observe that $S(X_t,Y_t,Z) = S(X_t,Y_t,Z_t)$, where $Z_t = Z\cap(x_t+t_t+\mc V_t)$.
	We shall now further decompose $S(X_t,Y_t,Z_t)$ so as to make the third set spread.
	By Lemma~\ref{lemma:uniformization_one_set} applied to $Z_t$,\footnote{Formally, we shift the set $Z_t$ so it lies in the linear subspace $\mc V_t$, then apply the lemma and shift back.} with parameters $r,\epsilon,\eta\alpha$,  we can find $T_t\in \N$, and for each $t'\in [T_t]$, a linear subspace $\mc V_{t,t'}\subseteq \mc V_t$ of dimension $\dim(\mc V_{t,t'})\geq \dim(\mc V_t)-r\epsilon^{-1}\log_2(1/(\eta\alpha))$, a point $z_{t,t'}\in \mc V_t/\mc V_{t,t'}$, and a subset $Z_{t,t'}\subseteq x_t+y_t+z_{t,t'}+\mc V_{t,t'}$ of size $\abs{Z_{t,t'}}\geq\eta \alpha\abs{\mc V_{t,t'}}$ that is $(r,\epsilon)$-algebraically spread.
	Also, we may write $Z_t$ as a disjoint union $Z_t=Z_{t,1}\cup\dots\cup Z_{t,T_t}\cup \tilde{Z_t}$, with $|\tilde{Z_t}|\leq \eta\alpha \abs{\mc V_t}$.
	
	Note that for all $t\in [T],t'\in [T_t]$, we have $S(X_t,Y_t,Z_{t,t'}) = \cup_{z\in \mc V_t/\mc V_{t,t'}} S( X_{t,t',z} , Y_{t,t',z} , Z_{t,t'})$, where $X_{t,t',z} = X_t\cap (x_t+z+\mc V_{t,t'}),\ Y_{t,t',z}=Y_t\cap (y_t+z+z_{t,t'}+\mc V_{t,t'})$; recall $Z_{t,t'}\subseteq x_t+y_t+z_{t,t'}+\mc V_{t,t'}$.
	
	Finally, we define the sets $X_i,Y_i,Z_i$ in the lemma statement as all $X_{t,t',z}, Y_{t,t',z}, Z_{t,t'}$, for $t\in [T], t'\in [T_t], z\in \mc V_t/\mc V_{t,t'}$ that satisfy \[\abs{S(X_{t,t',z}, Y_{t,t',z}, Z_{t,t'})} \geq \eta^2\alpha^2\kappa\cdot \abs{\mc V_{t,t'}}^2.\]
	Also, we define $\mc G$ to contain all $X_{t,t',z}, Y_{t,t',z}, Z_{t,t'}$ satisfying 
    \[\frac{\abs{X_{t,t',z}}}{\abs{\mc V_{t,t'}}}\geq (1-4\epsilon/10)\cdot \frac{\abs{X_t}}{\abs{\mc V_t}},\quad \frac{\abs{Y_{t,t',z}}}{\abs{\mc V_{t,t'}}}\geq (1-4\epsilon/10)\cdot \frac{\abs{Y_t}}{\abs{\mc V_t}}.\]
	We will show later (in the last item below) that the pieces in the set $\mc G$ also satisfy the condition $\abs{S(X_{t,t',z}, Y_{t,t',z}, Z_{t,t'})} \geq \eta^2\alpha^2\kappa\cdot \abs{\mc V_{t,t'}}^2$.
	Now, we verify each of the conclusions in the lemma statement:
	
	\begin{enumerate}
		\item We have for all $t,t'$ that
		\begin{align*}
			\dim(\mc V_{t,t'}) &\geq \dim(\mc V_t)-r\epsilon^{-1}\log_2(1/(\eta\alpha))
			\\&\geq \dim(\mc V)- O\brac{r_0\epsilon^{-2}\log_2(1/(\eta\alpha))^5}
			\\&\geq \dim(\mc V) - O\brac{r\epsilon^{-3}\log_2(1/(\eta\alpha))^6}.
		\end{align*}
		\item By construction, it holds that the sets $S(X_i,Y_i,Z_i)$ are disjoint subsets of $S(X,Y,Z)$. Let $\mc I$ to denote all $(t,t',z)$ such that $\abs{S(X_{t,t',z}, Y_{t,t',z}, Z_{t,t'})} < \eta^2\alpha^2\kappa\cdot \abs{\mc V_{t,t'}}^2$. Then, the quantity $\abs{S(X,Y,Z)}-\sum_i \abs{S(X_i,Y_i,Z_i)}$ equals
		\[ |\widetilde{XY}\cap \set{(x,y):x+y\in Z}| + \sum_{t} |S(X_t,Y_t,\tilde{Z_t})| + \sum_{(t,t',z)\in \mc I}\abs{S(X_{t,t',z}, Y_{t,t',z}, Z_{t,t'})} . \]
		We bound each of these one by one.
		\begin{enumerate}
			\item The first term is at most $|\widetilde{XY}| \leq \eta\alpha\abs{X}\abs{Y} \leq \eta\alpha\abs{\mc V}^2$.
			\item For any $t\in [T]$, we have that $X_t, Y_t$ are $(r_0,\epsilon/10)$-algebraically spread, and hence also $(r,\epsilon)$-algebraically spread. Thus, by Theorem~\ref{thm:algspread_conv}, we have $|S(X_t,Y_t,\tilde{Z_t})| \leq 1.1 \cdot \abs{X_t}\abs{Y_t}\cdot \eta\alpha$. Here, we used that $X_t, Y_t$ have density at least $\kappa$, and also that for this calculation we may assume that $\tilde{Z_t}$ has density at least $\kappa$ (which is at most $\eta\alpha/2$) by possibly adding more elements to $\tilde{Z}_t$. Hence, the second term is at most
			\[ \sum_{t\in [T]} |S(X_t,Y_t,\tilde{Z_t})|  \leq 1.1\cdot \eta\alpha\cdot  \sum_{t\in [T]}\abs{X_t}\abs{Y_t}  \leq 1.1\cdot \eta\alpha\cdot  \abs{X}\abs{Y} \leq 1.1\cdot \eta \alpha \abs{\mc V}^2.\]
			\item The third term is at most
			\begin{align*}
				\sum_{(t,t',z)\in \mc I}\abs{S(X_{t,t',z}, Y_{t,t',z}, Z_{t,t'})} &\leq \sum_{(t,t',z)\in \mc I}\eta^2\alpha^2\kappa\cdot \abs{\mc V_{t,t'}}^2
				\leq \sum_{\substack{t,t'\\z\in\mc V_t/\mc V_{t,t'}}}\eta^2\alpha^2\kappa\cdot\abs{\mc V_{t,t'}}^2
				\\&= \sum_{t,t'}\eta\alpha\kappa\cdot \brac{\eta\alpha\abs{\mc V_{t,t'}}}\cdot \abs{\mc V_{t}}
				\leq  \sum_{t,t'}\eta\alpha\kappa\cdot \abs{Z_{t,t'}}\cdot \abs{\mc V_{t}}
				\\&\leq \sum_{t}\eta\alpha \cdot \brac{\kappa\abs{\mc V_{t}}^2}
				\leq  \sum_{t}\eta\alpha\cdot \abs{X_t}\abs{Y_t} \leq \eta\alpha\abs{\mc V}^2.
			\end{align*}
		\end{enumerate}
		Combining, we get the bound $3.1\cdot \eta\alpha\abs{\mc V}^2 \leq  4\eta\abs{S(X,Y,Z)}$.
		
		\item By construction, we have $\abs{S(X_i,Y_i,Z_i)} \geq \eta^2\alpha^2 \kappa\abs{\mc V_i}^2 \geq 2^{-O(\log_2(1/(\eta\alpha))^2)}\cdot \abs{\mc V_i}^2$.
		
		\item Each $Z_{t,t'}$ is $(r,\epsilon)$-algebraically spread within $x_t+y_t+z_{t,t'}+\mc V_{t,t'}$.
		
		Now, consider any $X_{t,t',z}$ such that $\frac{\abs{X_{t,t',z}}}{\abs{\mc V_{t,t'}}}\geq (1-4\epsilon/10)\cdot \frac{\abs{X_t}}{\abs{\mc V_t}}$. We show that it is $(r,\epsilon)$-algebraically spread within $x_t+z+\mc V_{t,t'}$.
		For this, let $\mc W \subseteq  x_t+z+\mc V_{t,t'}$ be any affine subspace satisfying $\dim(\mc W) \geq \dim(\mc V_{t,t'})-r \geq \dim(\mc V_t)-r_0$.
		Then, as $X_t$ is $(r_0,\epsilon/10)$-algebraically spread within $x_t+\mc V_t$, we get
		\[ \frac{\abs{X_{t,t',z}\cap \mc W}}{\abs{\mc W}} = \frac{\abs{X_t\cap \mc W}}{\abs{\mc W}} \leq \brac{1+\frac{\epsilon}{10}}\cdot \frac{\abs{X_t}}{\abs{\mc V_t}} \leq \frac{\brac{1+\frac{\epsilon}{10}}}{\brac{1-\frac{4\epsilon}{10}}} \cdot  \frac{\abs{X_{t,t',z}}}{\abs{\mc V_{t,t'}}} \leq \brac{1+\epsilon}\cdot \frac{\abs{X_{t,t',z}}}{\abs{\mc V_{t,t'}}}.\]
		
		A similar argument shows that $Y_{t,t',z}$ is $(r,\epsilon)$-algebraically spread within $y_t+z+z_{t,t'}+\mc V_{t,t'}$, assuming that $\frac{\abs{Y_{t,t',z}}}{\abs{\mc V_{t,t'}}}\geq (1-4\epsilon/10)\cdot \frac{\abs{Y_t}}{\abs{\mc V_t}}$.
		
		\item Consider any piece $X_{t,t',z}, Y_{t,t',z}, Z_{t,t'}$ in $\mc G$. This satisfies \[ \frac{\abs{X_{t,t',z}}}{\abs{\mc V_{t,t'}}}\geq (1-4\epsilon/10)\cdot \frac{\abs{X_t}}{\abs{\mc V_t}} \geq \frac{\kappa}{2}, \quad \frac{\abs{Y_{t,t',z}}}{\abs{\mc V_{t,t'}}} \geq (1-4\epsilon/10)\cdot \frac{\abs{Y_t}}{\abs{\mc V_t}}\geq \frac{\kappa}{2},\quad \frac{\abs{Z_{t,t'}}}{\abs{\mc V_{t,t'}}} \geq\eta \alpha \geq \frac{\kappa}{2}.\]
		Since all three are $(r,\epsilon)$-algebraically spread, with $r\geq r',\epsilon\leq \epsilon'$, by Theorem~\ref{thm:algspread_conv}, 
		\begin{align}
			\abs{S(X_{t,t',z}, Y_{t,t',z}, Z_{t,t'})} &\geq \frac{9}{10}\cdot \frac{\abs{X_{t,t',z}}}{\abs{\mc V_{t,t'}}}\cdot \frac{\abs{Y_{t,t',z}}}{\abs{\mc V_{t,t'}}}\cdot \frac{\abs{Z_{t,t'}}}{\abs{\mc V_{t,t'}}}\cdot \abs{\mc V_{t,t'}}^2
			\nonumber\\&\geq  \frac{9}{10}\cdot \brac{1-\frac{4\epsilon}{10}}^2 \cdot \frac{\abs{X_t}\abs{Y_t}}{\abs{\mc V_t}^2}\cdot \frac{\abs{Z_{t,t'}}}{\abs{\mc V_{t,t'}}} \cdot  \abs{\mc V_{t,t'}}^2
			\nonumber\\&\geq  \frac{1}{2}\cdot \frac{\abs{X_t}\abs{Y_t}}{\abs{\mc V_t}^2}\cdot \abs{Z_{t,t'}}\abs{\mc V_{t,t'}} .\label{eqn:conv_lb}
		\end{align}
		In particular, the above is at least $1/2\cdot \kappa \cdot \eta\alpha \abs{\mc V_{t,t'}}^2 \geq \eta^2\alpha^2\kappa\cdot \abs{\mc V_{t,t'}}^2$, which shows that the pieces in $\mc G$ are a subset of the pieces $X_i,Y_i,Z_i$ that were not discarded.
		
		Now, fix any $t\in [T],t'\in [T_t]$.
		We show that \[ \Pr_{z\sim \mc V_t/\mc V_{t,t'}}\sqbrac{\frac{\abs{X_{t,t',z}}}{\abs{\mc V_{t,t'}}}\geq (1-4\epsilon/10)\cdot \frac{\abs{X_t}}{\abs{\mc V_t}}} \geq 4/5. \]
		For this, let $p$ denote the probability on the left hand side.
		Observe that for every $z\in \mc V_t/\mc V_{t,t'}$, as $X_t$ is $(r_0,\epsilon/10)$-algebraically spread within $x_t+\mc V_t$, we have the upper bound $\frac{\abs{X_{t,t',z}}}{\abs{\mc V_{t,t'}}}\leq (1+\epsilon/10)\cdot \frac{\abs{X_t}}{\abs{\mc V_t}}$.
		This implies that
		\[ \frac{\abs{X_t}}{\abs{\mc V_t}} =  \E_{z\sim \mc V_t/\mc V_{t,t'}}\sqbrac{\frac{\abs{X_{t,t',z}}}{\abs{\mc V_{t,t'}}}} \leq (1-p)\cdot (1-4\epsilon/10)\cdot \frac{\abs{X_t}}{\abs{\mc V_t}} + p\cdot (1+\epsilon/10)\cdot \frac{\abs{X_t}}{\abs{\mc V_t}},\]
		which gives $p\geq 4/5$.
	
		Thus, for any $t\in [T],t'\in [T_t]$, by the above, and by using a similar bound for $Y_{t,t',z}$,
		 \[ \Pr_{z\sim \mc V_t/\mc V_{t,t'}}\sqbrac{\frac{\abs{X_{t,t',z}}}{\abs{\mc V_{t,t'}}}\geq (1-4\epsilon/10)\cdot \frac{\abs{X_t}}{\abs{\mc V_t}},\quad \frac{\abs{Y_{t,t',z}}}{\abs{\mc V_{t,t'}}}\geq (1-4\epsilon/10)\cdot \frac{\abs{Y_t}}{\abs{\mc V_t}}} \geq 3/5. \]
		Combining with the earlier bound (Equation~\ref{eqn:conv_lb}), we get
		\begin{align*}
			\sum_{i\in \mc G} \abs{S(X_i,Y_i,Z_i)} &\geq \sum_{t\in [T],t'\in [T_t]}\frac{1}{2}\cdot \frac{\abs{X_t}\abs{Y_t}}{\abs{\mc V_t}^2}\cdot \abs{Z_{t,t'}}\abs{\mc V_{t,t'}} \cdot \brac{\frac{3}{5} \cdot \frac{\abs{\mc V_t}}{\abs{\mc V_{t,t'}}}}
			\\&= \sum_{t\in [T],t'\in [T_t]}\frac{3}{10}\cdot \frac{\abs{X_t}\abs{Y_t}}{\abs{\mc V_t}}\cdot \abs{Z_{t,t'}}
			\\&\geq \sum_{t\in [T]}\frac{3}{10}\cdot \frac{\abs{X_t}\abs{Y_t}}{\abs{\mc V_t}}\cdot \brac{\abs{Z_t}-\eta\alpha\abs{\mc V_t}}
			\\&= \sum_{t\in [T]}\frac{3}{10}\cdot \frac{\abs{X_t}\abs{Y_t}\abs{Z_t}}{\abs{\mc V_t}}  - \eta\alpha \sum_{t\in [T]}\abs{X_t}\abs{Y_t}
			\\&\geq \sum_{t\in [T]}\frac{3}{10}\cdot \frac{\abs{X_t}\abs{Y_t}\abs{Z_t}}{\abs{\mc V_t}}\cdot \ind\sqbrac{\abs{Z_t}\geq \eta\alpha\abs{\mc V_t}}  - \eta\alpha \abs{\mc V}^2.
		\end{align*}
		Now, as $X_t,Y_t$ are $(r_0,\epsilon/10)$-algebraically spread within $x_t+\mc V_t, y_t+\mc V_t$ respectively, by the choice of $r',\epsilon'$, and by Theorem~\ref{thm:algspread_conv}, we have $\abs{S(X_t,Y_t,Z_t)} \leq 1.1 \cdot \frac{\abs{X_t}\abs{Y_t}\abs{Z_t}}{\abs{\mc V_t}}$ whenever $\abs{Z_t}\geq \eta\alpha\abs{\mc V_t} \geq \kappa \abs{\mc V_t}$.
        Hence,
		\begin{align*}
			\sum_{i\in \mc G} \abs{S(X_i,Y_i,Z_i)} &\geq  \sum_{t\in [T]} \frac{3}{10}\cdot \frac{1}{1.1}\cdot  \abs{S(X_t,Y_t,Z_t)} \cdot \ind\sqbrac{\abs{Z_t}\geq \eta\alpha\abs{\mc V_t}}- \eta\alpha\abs{\mc V}^2
			\\&\geq 0.25\cdot \sum_{t\in [T]}\abs{S(X_t,Y_t,Z_t)} - \sum_{t\in [T]}\abs{S(X_t,Y_t,Z_t)}\cdot \ind\sqbrac{\abs{Z_t}<\eta\alpha\abs{\mc V_t}} -\eta\alpha\abs{\mc V}^2
			\\&\geq 0.25\cdot \abs{S(X,Y,Z)}-\eta\alpha\abs{\mc V}^2- 1.1\cdot \sum_{t\in [T]}\abs{X_t}\abs{Y_t}\eta\alpha -\eta\alpha\abs{\mc V}^2.
			\\&\geq 0.25\cdot \abs{S(X,Y,Z)}-3.1\eta\alpha\abs{\mc V}^2
			\\&= (0.25-3.1\eta)\cdot  \abs{S(X,Y,Z)}  \geq 0.1\cdot \abs{S(X,Y,Z)}.
		\end{align*}
		For the middle term $\abs{S(X_t,Y_t,Z_t)}\cdot \ind\sqbrac{\abs{Z_t}<\eta\alpha}$, we used that $X_t,Y_t$ are $(r_0,\epsilon/10)$-algebraically spread within $x_t+\mc V_t, y_t+\mc V_t$ respectively, and also for the purpose of calculating an upper bound assumed that $\abs{Z_t}/\abs{\mc V_t} \geq \kappa$ which can be obtained by possibly adding more elements to $Z_t$.
		\qedhere
	\end{enumerate}
\end{proof}

We write down a simpler version of the above lemma, with the assumption on $r,\epsilon$ removed:

\begin{corollary}\label{corr:one_round_part}
	Let $r\in \N$, $\epsilon,\eta\in (1/10)$, and let $X,Y,Z\subseteq \mc V$ for a linear subspace $\mc V\subseteq \F_2^n$.
	Define $\alpha = \abs{S(X,Y,Z)}/\abs{\mc V}^2$.
	Then, there exists $T\in \N$	, a set $\mc G\subseteq [T]$, and for each $i\in [T]$, a linear subspace $\mc V_i\subseteq \mc V$, points $x_i,y_i\in \mc V/\mc V_i$, and subsets $X_i\subseteq x_i+\mc V_i, Y_i\subseteq y_i+\mc V_i, Z_i\subseteq x_i+y_i+\mc V_i$, such that
	\begin{enumerate}
		\item $\dim(\mc V_i)\geq \dim(\mc V) - O\brac{r\epsilon^{-3}\log_2(1/(\eta\alpha))^6 + \epsilon^{-3}\log_2(1/(\eta\alpha))^{22}}$ for all $i\in [T]$.
		\item $S(X_1,Y_1,Z_1),\dots,S(X_T,Y_T,Z_T)$ are disjoint subsets of $S(X,Y,Z)$ such that 
			\[ \abs{ S(X,Y,Z)\setminus \cup_{i=1}^T S(X_i,Y_i,Z_i)} \leq  \eta\abs{S(X,Y,Z)}. \]
		\item $\abs{S(X_i,Y_i,Z_i)} \geq 2^{-O(\log_2(1/(\eta\alpha))^2)}\cdot \abs{\mc V_i}^2$ for all $i\in [T]$.
		\item For all $i\in \mc G$, the sets $X_i,Y_i,Z_i$ are $(r,\epsilon)$-algebraically spread within $x_i+\mc V_i,y_i+\mc V_i,x_i+y_i+\mc V_i$ respectively.
		\item $\sum_{i\in \mc G} \abs{S(X_i,Y_i,Z_i)} \geq \frac{1}{10}\cdot \abs{S(X,Y,Z)}$.
	\end{enumerate}
\end{corollary}
\begin{proof}
	We plug in the parameter $\eta/5$ in Lemma~\ref{lemma:one_round_part} to get $r' = O(\log_2(1/(\eta\alpha))^{16})$, and $\epsilon' = \Omega(1)$.
	Then, we use the lemma with parameters $r_0 = r+r'$ and $\epsilon_0 = \epsilon \cdot \epsilon'$, noting that any set which is $(r_0,\epsilon_0)$-algebraically spread is also $(r,\epsilon)$-algebraically spread.
\end{proof}

Finally, we complete the main result of this section:
\begin{proof}[Proof of Proposition~\ref{prop:uniformization}]
We perform the following recursive process, given sets $X,Y,Z\subseteq \mc V$:
Let $T\in \N$, $\mc G\subseteq [T]$, and subsets $X_i\subseteq x_i+\mc V_i, Y_i\subseteq y_i+\mc V_i, Z_i\subseteq x_i+y_i+\mc V_i$ be as in Corollary~\ref{corr:one_round_part}, with parameters $r,\frac{\epsilon}{10},\frac{\eta^2}{100}$.
Now, for each $i\in [T]\setminus \mc G$, we invoke Corollary~\ref{corr:one_round_part} recursively on $X_i,Y_i,Z_i$,\footnote{Formally, we shift the sets so they all in the linear subspace $\mc V_i$, then apply the corollary and shift back.} with parameters $r,\frac{\epsilon}{10},\frac{\eta^2}{100}$.
We terminate when the recursion depth is $L = \lceil20\log_2(1/\eta)\rceil$, and throw away all the remaining pieces at this point.

Note that for any $X_i,Y_i,Z_i$ lying in affine shifts of a linear subspace $\mc V_i\subseteq \mc V$, that the algorithm sees at some point but does not throw away, it holds that 
\[ \frac{\abs{S(X_i,Y_i,Z_i)}}{\abs{\mc V_i}^2} \geq 2^{-(\log_2(4/\alpha))^{\poly(1/\eta)}}.\]
This follows by Corollary~\ref{corr:one_round_part}, as the measure reduces from $\beta$ to $2^{-O(\log_2(1/(\eta\beta))^2)}$ in a single step, and the number of steps is $L=O(\log_2(1/\eta))$.
With this observation, we verify each of the conclusions in the statement of the proposition:
\begin{enumerate}
	\item At each step, the dimension reduces by at most $r\epsilon^{-3}\log_2(4/\alpha)^{\poly(1/\eta)}$, and the number of steps is $L=O(\log_2(1/\eta))$.
	\item The disjointness of the pieces follows by construction.
	
	At each recursive layer, the total size of pieces thrown out is at most $\frac{\eta^2}{100}\cdot \abs{S(X,Y,Z)}$.
	
	Additionally, at any layer the total size of the pieces being recursed on is at most $9/10$ times the total size of pieces in the previous layer; in particular, the total size of pieces at layer $L$ is at most $(9/10)^L\cdot \abs{S(X,Y,Z)} \leq \eta^2\abs{S(X,Y,Z)}$.	Hence, the total fraction of pieces thrown out is at most
	\[ \ \frac{\eta^2}{100} \cdot 21\log_2(1/\eta) + \eta^2 \leq \eta.\]
	\item The size bound for $X_i,Y_i,Z_i$ is as above, and spreadness follows by Corollary~\ref{corr:one_round_part}. \qedhere
\end{enumerate}
\end{proof} 
\section{Uniform Square Covers}\label{sec:unif_sq_covers}

In this section, we show that any \emph{diagonal-product} set (see Definition~\ref{defn:diag_prod}) composed of algebraically spread sets (see Definition~\ref{defn:alg_spread}) is \emph{uniformly covered} by \emph{squares} (see Definition~\ref{defn:square}).
In Section~\ref{sec:comb_spread}, we introduce a combinatorial notion of spreadness that will be useful for us; in Section~\ref{sec:algspread_imp_combspread}, we show that algebraically spread sets satisfy this definition (in a specific way); finally, in Section~\ref{sec:sq_cover}, we prove Proposition~\ref{prop:sq_cover}, our square covering result.

\subsection{Combinatorial Spreadness and Graph Counts}\label{sec:comb_spread}

We state some useful definitions from~\cite{KLM24}:

\begin{definition}[Combinatorial Spreadness]
Let $X, Y$ be finite sets, and let $f:X\times Y\to \set{0,1}$.
We say that $f$ is $(r,\epsilon)$-combinatorially spread if for every $S\subseteq X,\ T\subseteq Y$ satisfying $\abs{S\times T}\geq 2^{-r}\abs{X\times Y}$, it holds that
\[ \E_{(x,y)\sim S\times T}[f(x,y)] \leq (1+\epsilon)\E[f],\]
where $\E[f] = \E_{(x,y)\sim X\times Y}[f(x,y)]$.
\end{definition}

\begin{definition}[Lower Bounded Left-Marginals]
	Let $X, Y$ be finite sets, and let $f:X\times Y\to \set{0,1}$.
	We say that $f$ has $(r,\epsilon)$-lower bounded left-marginals if
	\[ \Pr_{x\sim X}\sqbrac{ \E_{y\sim Y}[f(x,y)] \leq (1-\epsilon)\cdot \E[f] } \leq 2^{-r} .\]
\end{definition}

We state a result about graph counts in spread sets:

\begin{theorem}[{\cite[Theorem 2.1]{FHHK24}}]\label{thm:graph_count}
For any $k\in \N,\ \epsilon \in (0,1)$, there exists sufficiently small $\gamma = \gamma(\epsilon, k) > 0$, and sufficiently large $C = C(\epsilon, k) \in \N$, such that the following holds:

Let $H = ([k], E)$ be an oriented graph, where $(i,j)\in E$ implies $i<j$.
Let $X_1,\dots,X_k$ be finite sets; let $d\geq 1$, and for each $(i,j)\in E$, let $f_{ij}:X_i\times X_j\to \set{0,1}$ be a function satisfying:
\begin{enumerate}
	\item $\E[f_{ij}] \geq 2^{-d}$,
	\item $f_{ij}$ is $(Cd^2,\gamma)$-combinatorially spread, and
	\item $f_{ij}$ has $(Cd, \gamma)$-lower bounded left-marginals.
\end{enumerate}
Then, it holds that
\[ \abs{ \E_{x_1\sim X_1,\dots,x_k\sim X_k}\sqbrac{\prod_{(i,j)\in E}f_{ij}(x_i,x_j)}- \prod_{(i,j)\in E} \E[f_{ij}]} \leq \epsilon\cdot \prod_{(i,j)\in E} \E[f_{ij}] .\]
\end{theorem}

\subsection{Algebraic Spreadness implies Graph Counts}\label{sec:algspread_imp_combspread}

We show that a certain function corresponding to algebraically spread sets is both combinatorially spread and has lower bounded left-marginals.

\begin{lemma}\label{lemma:algspread_imp_combspread}
	Let $d,r\geq 1,\epsilon\in (0,1/4)$.
	Then, there exists a sufficiently large integer $s = (r^8+d^8)\cdot \epsilon^{-O(1)}$, and a sufficiently small $\delta = \Omega(\epsilon)$, such that the following holds:
	
	Let $X,Y,Z \subseteq \F_2^n$ be subsets each of density at least $2^{-d}$, and such that $Y,Z$ are $(s,\delta)$-algebraically spread.	
	Then, the function $f:X\times Y\to \set{0,1}$, given by $f(x,y)=\ind\sqbrac{x+y\in Z}$, satisfies:
	\begin{enumerate}
		\item $ \abs{\E[f] - \alpha_Z} \leq \epsilon\cdot \alpha_Z$, where $\alpha_Z = \abs{Z}/2^n$,
		\item $f$ is $(r,\epsilon)$-combinatorially spread, and
		\item $f$ has $(r, \epsilon)$-lower bounded left-marginals.
	\end{enumerate}
\end{lemma}
\begin{proof}
	Let $s,\delta$ be such that Theorem~\ref{thm:algspread_conv} holds with parameters $r+d, \epsilon/4$.
	
	First, we have
	\[ \E[f] = \Pr_{x\sim X, y\sim Y}[x+y\in Z] =\ip{\varphi_X *\varphi_Y, \ind_Z}  = \alpha_Z\cdot \ip{\varphi_X *\varphi_Y, \varphi_Z}. \]
	Hence, by Theorem~\ref{thm:algspread_conv}\ref{item:km_two_spread}, it holds that
	\[ \abs{\E[f]-\alpha_Z} \leq \frac{\epsilon}{4}\cdot \alpha_Z .\]
	
	Next, we show combinatorial spreadness.
	Let $S\subseteq X,\ T\subseteq Y$ be such that $\abs{S\times T} \geq 2^{-r}\abs{X\times Y}$.
	We can write
	\[ \E_{(x,y)\sim S\times T}[f(x,y)] = \alpha_Z\cdot \ip{\varphi_S *\varphi_T, \varphi_Z}.\]
	Since $\abs{X} \geq 2^{-d}\cdot 2^n$, we have $\abs{S}\cdot \abs{Y} \geq \abs{S\times T} \geq 2^{-r}\cdot \abs{X\times Y} \geq 2^{-r}\cdot 2^{n-d} \cdot \abs{Y}$, and so $\abs{S}\geq 2^{-(r+d)}\cdot 2^n$.
    Similarly, $\abs{T} \geq 2^{-(r+d)}\cdot 2^n$, and hence, by Theorem~\ref{thm:algspread_conv}\ref{item:km_one_spread}, we get
	\[ \E_{(x,y)\sim S\times T}[f(x,y)] \leq \brac{1+\frac{\epsilon}{4}}\alpha_Z \leq \brac{1+\frac{\epsilon}{4}}\cdot \brac{1-\frac{\epsilon}{4}}^{-1}\E[f] \leq (1+\epsilon)\cdot \E[f].\]
	
	Finally, we show lower bounded left-marginals.
	Define $A\subseteq X$ by \[ A = \set{x\in X: \E_{y\sim Y}[f(x,y)] \leq (1-\epsilon)\cdot \E[f] }. \]
	Suppose, for the sake of contradiction, that $|A| > 2^{-r} |X|$; then, $|A|\geq 2^{-(r+d)}\cdot 2^n$.
	By the definition of $A$, we have
	\[ \E_{x\sim A, y\sim Y}[f(x,y)] \leq (1-\epsilon)\E[f].\]
	On the other hand, by Theorem~\ref{thm:algspread_conv}\ref{item:km_two_spread}, we have
	\[ \E_{x\sim A, y\sim Y}[f(x,y)] = \alpha_Z\cdot \ip{\varphi_A *\varphi_Y, \varphi_Z} \geq  \brac{1-\frac{\epsilon}{4}}\alpha_Z \geq \brac{1-\frac{\epsilon}{4}}\cdot \brac{1+\frac{\epsilon}{4}}^{-1}\E[f] > (1-\epsilon)\E[f],\]
	which is a contradiction.
\end{proof}

\subsection{Squares inside Spread Diagonal-Product Sets}\label{sec:sq_cover}

In this subsection, we prove our square covering result.
We start by proving the following counting lemma:

\begin{lemma}\label{lemma:spread_imp_l1_l2}
	Let $\epsilon\in (0,1/4)$ be a constant, and let $d\geq 1$.
	Then, there exists a sufficiently large integer $r = O_{\epsilon}(d^{16})$ and sufficiently small $\delta=\delta(\epsilon)>0$ such that the following holds:

	Let $X,Y,Z\subseteq \F_2^n$ be each of density at least $2^{-d}$, and such that all of $X,Y,Z$ are $(r,\delta)$-algebraically spread.
	Let $S = S(X,Y,Z)$ be as in Definition~\ref{defn:diag_prod}, and let $\Gamma:\F_2^n\times \F_2^n\to \R$ be the function mapping $(x,y)$ to the normalized number of squares (see Definition~\ref{defn:square}) in $S$ containing $(x,y)$, i.e., 
	\[ \Gamma(x,y) = \E_{w\sim \F_2^n}\sqbrac{ \ind\sqbrac{s_{x,y,w}\subseteq S}}. \]
	Define $\alpha_X=\E[\ind_X], \alpha_Y=\E[\ind_Y], \alpha_Z=\E[\ind_Z]$.
	Then, it holds that
	\begin{enumerate}
		\item $\abs{\abs{S}\cdot 2^{-2n}-\alpha_X\alpha_Y\alpha_Z} \leq \epsilon\cdot \alpha_X\alpha_Y\alpha_Z$. 
		\item $\abs{\Norm{\Gamma}_1-\alpha_X^2\alpha_Y^2\alpha_Z^2} \leq \epsilon\cdot \alpha_X^2\alpha_Y^2\alpha_Z^2$.
		\item $\abs{\Norm{\Gamma}_2^2-\alpha_X^3\alpha_Y^3\alpha_Z^3} \leq \epsilon\cdot \alpha_X^3\alpha_Y^3\alpha_Z^3$. 
	\end{enumerate}	
\end{lemma}
\begin{proof}
	Let $f:X\times Y\to \set{0,1},\ g:Y\times Z\to \set{0,1},\ h:X\times Z\to \set{0,1} $ be given by \[ f(x,y)=\ind[x+y\in Z],\quad g(y,z)=\ind[y+z\in X], \quad h(x,z)=\ind[x+z\in Y]. \]
	Let $\gamma=\gamma(\frac{\epsilon}{10}, 4)\leq \frac{\epsilon}{10},\ C=C(\frac{\epsilon}{10}, 4)$ be as in Theorem~\ref{thm:graph_count}.
	Then, by Lemma~\ref{lemma:algspread_imp_combspread} and the choice of parameters $(r,\delta)$, we may assume that $f$ (and similarly $g,h$) is such that $\abs{\E[f]-\alpha_Z}\leq \gamma\alpha_Z \leq \frac{\epsilon}{10}\alpha_Z$, and $f$ is $(Cd^2,\gamma)$-combinatorially spread, and $f$ has $(Cd, \gamma)$-lower bounded left-marginals.
	
	First, observe that 
	\[ |S|\cdot 2^{-2n} = \E_{x,y\sim \F_2^n}\sqbrac{\ind\sqbrac{x\in X, y\in Y, x+y\in Z}} = \alpha_X\alpha_Y \Pr_{x\sim X, y\sim Y}[x+y\in Z] = \alpha_X\alpha_Y\E[f]. \]
	The result now follows, as $\abs{\E[f]-\alpha_Z}\leq \gamma\alpha_Z \leq \frac{\epsilon}{10} \alpha_Z$.
	
	Next, we note that 
	\begin{align*}
		\Norm{\Gamma}_1 &= \E_{x,y,w\sim \F_2^n}\sqbrac{\ind\sqbrac{s_{x,y,w}\subseteq S}} 
		\\&= \E_{x,y,w\sim \F_2^n}\sqbrac{\ind_X(x)\ind_X(x+w)\ind_Y(y)\ind_Y(y+w)\ind_Z(x+y)\ind_Z(x+y+w)}
		\\&= \E_{x,y,z\sim \F_2^n}\sqbrac{\ind_X(x)\ind_X(y+z)\ind_Y(y)\ind_Y(x+z)\ind_Z(x+y)\ind_Z(z)} \quad (\text{replace } w= x+y+z)
		\\&= \alpha_X\alpha_Y\alpha_Z \E_{x\sim X, y\sim Y, z\sim Z}\sqbrac{f(x,y)g(y,z)h(x,z)}.
	\end{align*}
	Now, by Theorem~\ref{thm:graph_count} (with the oriented graph $\set{(1,2),(1,3), (2,3)}$, with $1,2,3$ labelled by $X,Y,Z$ respectively), we get
	\begin{align*}
		(1-\epsilon)\cdot \alpha_X\alpha_Y\alpha_Z &\leq \brac{1-\frac{\epsilon}{10}}^4 \cdot \alpha_X\alpha_Y\alpha_Z
		\\&\leq \brac{1-\frac{\epsilon}{10}}\cdot \E[f]\E[g]\E[h]
		\\& \leq \E_{x\sim X, y\sim Y, z\sim Z}\sqbrac{f(x,y)g(y,z)h(x,z)}
		\\& \leq \brac{1+\frac{\epsilon}{10}}\cdot \E[f]\E[g]\E[h]
		\\&\leq \brac{1+\frac{\epsilon}{10}}^4 \cdot \alpha_X\alpha_Y\alpha_Z \leq (1+\epsilon)\cdot \alpha_X\alpha_Y\alpha_Z,
	\end{align*} 
	and hence 
	\[ \abs{\Norm{\Gamma}_1-\alpha_X^2\alpha_Y^2\alpha_Z^2} \leq \epsilon\cdot \alpha_X^2\alpha_Y^2\alpha_Z^2.\]
	
	Finally, we can write
	\begin{align*}
		\Norm{\Gamma}_2^2 &= \E_{x,y,w,w'\sim \F_2^n}\sqbrac{\ind\sqbrac{s_{x,y,w}\subseteq S, s_{x,y,w'}\subseteq S}} 
		\\&= \E_{x,y,z,z'\sim \F_2^n}\sqbrac{\ind\sqbrac{s_{x,y,x+y+z}\subseteq S, s_{x,y,x+y+z'}\subseteq S}} \quad (\text{replace } w= x+y+z, w'=x+y+z')
		\\&= \E_{x,y,z,z'\sim \F_2^n}\sqbrac{\ind_X(x)\ind_X(y+z)\ind_X(y+z')\ind_Y(y)\ind_Y(x+z)\ind_Y(x+z')\ind_Z(x+y)\ind_Z(z)\ind_Z(z')} 
		\\&= \alpha_X\alpha_Y\alpha_Z^2 \E_{x\sim X, y\sim Y, z,z'\sim Z}\sqbrac{f(x,y)g(y,z)g(y,z')h(x,z)h(x,z')}.
	\end{align*}
	Now, by Theorem~\ref{thm:graph_count} (with the oriented graph $\set{(1,2), (1,3), (1,4), (2,3), (2,4)}$, with $1,2,3,4$ labelled by $X,Y,Z,Z$ respectively), we get  
	\begin{align*}
		(1-\epsilon)\cdot \alpha_X^2\alpha_Y^2\alpha_Z &\leq \brac{1-\frac{\epsilon}{10}}^6 \cdot \alpha_X^2\alpha_Y^2\alpha_Z
		\\&\leq \brac{1-\frac{\epsilon}{10}}\cdot \E[f]\E[g]^2\E[h]^2
		\\& \leq \E_{x\sim X, y\sim Y, z,z'\sim Z}\sqbrac{f(x,y)g(y,z)g(y,z')h(x,z)h(x,z')}
		\\& \leq \brac{1+\frac{\epsilon}{10}}\cdot \E[f]\E[g]^2\E[h]^2
		\\&\leq \brac{1+\frac{\epsilon}{10}}^6 \cdot \alpha_X^2\alpha_Y^2\alpha_Z \leq (1+\epsilon)\cdot \alpha_X^2\alpha_Y^2\alpha_Z,
	\end{align*}
	and hence
	\[\abs{\Norm{\Gamma}_2^2-\alpha_X^3\alpha_Y^3\alpha_Z^3} \leq \epsilon\cdot \alpha_X^3\alpha_Y^3\alpha_Z^3. \qedhere\]
\end{proof}

With the above, we are ready to prove the main result of this section.

\begin{proposition}\label{prop:sq_cover}
	Let $\epsilon\in (0,1/4)$ be a constant, and let $d\geq 1$.
	Then, there exists a sufficiently large integer $r = O_{\epsilon}(d^{16})$ and sufficiently small $\delta=\delta(\epsilon)>0$ such that the following holds:

	Let $X,Y,Z\subseteq \F_2^n$, be each of density at least $2^{-d}$, and such that all of $X,Y,Z$ are $(r,\delta)$-algebraically spread.
	Let $S = S(X,Y,Z)$ be as in Definition~\ref{defn:diag_prod}, and let $\mc T$ denote the set of all squares in $S$ (see Definition~\ref{defn:square}); formally,
	\[ \mc T = \set{(x,y,w)\in (\F_2^n)^3 : s_{x,y,w}\subseteq S}.  \]
	Define $\alpha_X=\E[\ind_X], \alpha_Y=\E[\ind_Y], \alpha_Z=\E[\ind_Z]$.
	Then, it holds that
	\begin{enumerate}
		\item $\abs{\abs{S}\cdot 2^{-2n}-\alpha_X\alpha_Y\alpha_Z} \leq \epsilon\cdot \alpha_X\alpha_Y\alpha_Z$. 
		\item $\abs{\abs{\mc T}\cdot 2^{-3n}-\alpha_X^2\alpha_Y^2\alpha_Z^2} \leq \epsilon \cdot \alpha_X^2\alpha_Y^2\alpha_Z^2$.
		\item Let $\mu$ be the distribution on $S$ obtained as follows: Pick a random square in $S$, by picking uniformly at random $(x,y,w)\sim \mc T$; then, output a uniformly random element from the set $s_{x,y,w}$. It holds that 
			\[ \Norm{\mu - U_S}_1 \leq \epsilon, \]
			where $U_S$ denotes the uniform distribution over $S$.
	\end{enumerate}
\end{proposition}
\begin{proof}
	Let $(r,\delta)$ be such that Lemma~\ref{lemma:spread_imp_l1_l2} holds with parameters $d$ and $\epsilon^2/5:=\gamma$, and 
	let $\Gamma$ be the function as in the statement of Lemma~\ref{lemma:spread_imp_l1_l2}.
	By Lemma~\ref{lemma:spread_imp_l1_l2}, it directly holds that \[ \abs{\abs{S}\cdot 2^{-2n}-\alpha_X\alpha_Y\alpha_Z} \leq \gamma\cdot \alpha_X\alpha_Y\alpha_Z \leq \epsilon\cdot \alpha_X\alpha_Y\alpha_Z. \]
	
	Now, observe that 
	\[ \abs{\mc T} = \sum_{x,y,w\in \F_2^n}\ind\sqbrac{s_{x,y,w}\subseteq S} = 2^{3n}\Norm{\Gamma}_1, \]
	and hence by Lemma~\ref{lemma:spread_imp_l1_l2},
	\[ \abs{\abs{\mc T}\cdot 2^{-3n}-\alpha_X^2\alpha_Y^2\alpha_Z^2}  \leq \gamma \cdot \alpha_X^2\alpha_Y^2\alpha_Z^2\leq \epsilon \cdot \alpha_X^2\alpha_Y^2\alpha_Z^2.\]
	
	Finally, we recall Remark~\ref{remark:sq_rep}: every square  $s_{x,y,w}\subseteq S$ with $w\not=0$ occurs exactly 4 times in $\mc T$, as $s_{x,y,w}=s_{x+w,y,w}=s_{x,y+w,w}=s_{x+w,y+w,w}$.
	Also, when $w=0$, we have $s_{x,y,0} = \set{(x,y)}$.
	Thus, for any $(x,y)\in S$,
	\begin{equation}\label{eqn:square_count}
		\mu[(x,y)] = \frac{1}{\abs{\mc T}}\cdot\brac{ \sum_{w\in\F_2^n, w\not=0}4\cdot \ind\sqbrac{s_{x,y,w}\subseteq S}\cdot \frac{1}{4} + \ind\sqbrac{s_{x,y,0}\subseteq S}\cdot 1 } = \frac{\Gamma(x,y)\cdot 2^n}{\abs{\mc T}}.
	\end{equation}
	Hence, we can compute the $\ell_2$-norm of $\mu$ as
	\begin{align*}
		\Norm{\mu}_2^2 &= \sum_{(x,y)\in S} \mu[(x,y)]^2
		\\&= \frac{2^{2n}}{\abs{\mc T}^2}\cdot\sum_{(x,y)\in S}\Gamma(x,y)^2
		\\&= \frac{2^{2n}}{\abs{\mc T}^2}\cdot\sum_{x,y\in \F_2^n}\Gamma(x,y)^2 \qquad\qquad\qquad(\text{as }\Gamma(x,y)=0 \text{ for } (x,y)\not\in S)
		\\&= \frac{2^{4n}}{\abs{\mc T}^2}\cdot \Norm{\Gamma}_2^2.
	\end{align*}
	Now, by Lemma~\ref{lemma:spread_imp_l1_l2}, we obtain
	\[ \Norm{\mu}_2^2 \leq \frac{2^{4n}}{\abs{\mc T}^2}\cdot (1+\gamma)\cdot \alpha_X^3\alpha_Y^3\alpha_Z^3,\]
	and so,
	\begin{align*}
		\Norm{\mu}_2^2 \cdot \abs{S} &\leq 2^{4n}\cdot (1+\gamma)\cdot \alpha_X^3\alpha_Y^3\alpha_Z^3 \cdot \frac{\abs{S}}{\abs{\mc T}^2}
		\\&\leq 2^{4n}\cdot (1+\gamma)\cdot \alpha_X^3\alpha_Y^3\alpha_Z^3 \cdot \frac{(1+\gamma)\cdot \alpha_X\alpha_Y\alpha_Z\cdot 2^{2n}}{(1-\gamma)^2\cdot \alpha_X^4\alpha_Y^4\alpha_Z^4\cdot 2^{6n}}
		\\&= \frac{(1+\gamma)^2}{(1-\gamma)^2}\leq 1+5\gamma. \qquad\qquad (\text{as }\gamma\leq 0.1)
	\end{align*}
	Finally, by Cauchy-Schwarz, we get
	\[ \Norm{\mu-U_S}_1 \leq \brac{|S|\cdot \Norm{\mu-U_S}_2^2}^{1/2} = \brac{|S|\cdot \brac{\Norm{\mu}_2^2 - \frac{1}{\abs{S}}}}^{1/2} \leq \sqrt{5\gamma} = \epsilon. \qedhere\]
\end{proof}


\section{Parallel Repetition for Games with the GHZ Query Support}\label{sec:ghz_parrep}

This section is devoted to the proof of Theorem~\ref{thm:intro_main}.
By Lemma~\ref{lemma:wlog_unif_dist}, it suffices to consider the case where $Q$ is the uniform distribution over $\supp(Q)$.
We prove the following:

\begin{theorem}\label{thm:ghz_unif_parrep}
	Let $\mc G = (\mc X\times \mc Y\times \mc Z,\ \mc A\times \mc B\times \mc C,\ Q, V_{pred})$ be any 3-player game with $\mc X = \mc Y = \mc Z = \F_2$, and with $Q$ the uniform distribution over 
	\[ \supp(Q) =  \set{(x,y,z)\in \F_2^3:x+y+z=0},\]
	and such that $\val(\mc G)<1$.
	Then, for all sufficiently large $n$,\footref{footnote:suff_large_n} it holds that
	\[\val(\mc G^{\otimes n}) \leq \exp\brac{-n^{c}},\]
	where $c>0$ is an absolute constant.
\end{theorem}

The remainder of this section is devoted to the proof of the above theorem.
Towards this, we fix a 3-player game $\mc G$ as in the statement of the theorem; in particular, we must have that $\val(\mc G) \leq 3/4$.

The main step of the proof is to prove that the game $\mc G^{\otimes n}$ has a hard coordinate when the inputs to the players are conditioned on a product event with sufficiently large measure.
We carry this out in two parts: first, in Section~\ref{sec:hard_coor_assuming_spread}, we prove this statement assuming that each of the three sets in the product event are well spread (as in Definition~\ref{defn:alg_spread}) inside some large subspace of $\F_2^n$; then, in Section~\ref{sec:hard_coor}, we prove the statement for general product events, using our uniformization strategy.


\subsection{Hard Coordinate under Spread Product Events}\label{sec:hard_coor_assuming_spread}

We consider the game $\mc G^{\otimes n} = (\mc X^n\times \mc Y^n\times \mc Z^n,\ \mc A^n\times \mc B^n\times \mc C^n,\ Q^{\otimes n}, V_{pred}^{\otimes n})$ for some sufficiently large $n\in \N$.
Fix any strategies for the $3$ players in this game, and for each $i\in [n]$, let $\Win_i\subseteq \mc X^n\times \mc Y^n\times \mc Z^n$ be the event that this strategy wins the $i$\textsuperscript{th} coordinate of the game.

\begin{lemma}\label{lemma:hard_coor_assuming_spread}
	For any $1\leq d\leq o(n)$, and constant $\epsilon\in(0,1/4)$, there exists a (sufficiently large) integer $r = O_{\epsilon}(d^{16})$, and a (sufficiently small) constant $0<\delta=\delta(\epsilon)<1$, such that the following holds:
	
	Let $\mc V\subseteq \F_2^n$ be a linear subspace of dimension $\dim(\mc V)\geq n-o(n)$, and let $E,F,G\subseteq \mc V$ be such that each of them is $(r,\delta)$-algebraically spread within $\mc V$ (see Definition~\ref{defn:alg_spread}), and such that each of the densities $\frac{\abs{E}}{\abs{\mc V}}, \frac{\abs{F}}{\abs{\mc V}},\frac{\abs{G}}{\abs{\mc V}} \geq 2^{-d}$. 
	Then,
	\[ \E_{i\sim [n]}\sqbrac{\Pr_{Q^{\otimes n}}[\Win_i \mid E\times F\times G]} \leq \frac{1}{2}+\frac{\val(\mc G)}{2}+\epsilon. \]
\end{lemma} 

The remainder of this subsection is devoted to the proof of the Lemma~\ref{lemma:hard_coor_assuming_spread}.
We fix some $1\leq d\leq o(n)$, and constant $\epsilon\in(0,1/4)$, and let $r=O_{\epsilon}(d^{16}),\ \delta=\delta(\epsilon)>0$ be so that Proposition~\ref{prop:sq_cover} holds with the parameters $d,\frac{\epsilon}{2}$.
Let $\mc V\subseteq \F_2^n$ be a linear subspace of dimension $\dim(\mc V)\geq n-o(n)$, and let $E,F,G\subseteq \mc V$ be each $(r,\delta)$-algebraically spread within $\mc V$, and such that each of the densities $\alpha_E=\frac{\abs{E}}{\abs{\mc V}},\ \alpha_F=\frac{\abs{F}}{\abs{\mc V}},\ \alpha_G=\frac{\abs{G}}{\abs{\mc V}}$ is at least $2^{-d}$. 

We start by proving a lemma that demonstrates the usefulness of squares. It says that the no strategy for the game $\mc G^{\otimes n}$ can win on non-trivial coordinates of squares.

\begin{definition}[Non-trivial coordinates]
	For a square $s = s_{x,y,w}\subseteq \F_2^n\times \F_2^n$, as in Definition~\ref{defn:square}, we define its set of \emph{non-trivial coordinates} as 
	\[\mc I_{s} = \set{i\in [n]: w_i\not=0}.\]
	This is well-defined since $w$ is uniquely determined by the square $s$ (see Remark~\ref{remark:sq_rep}).
\end{definition}

\begin{lemma}\label{lemma:non_triv_in_extens_hard}
	Let $s$ be a square, and let $i\in \mc I_{s}$ be a non-trivial coordinate in $s$.
	Then,
	\[ \Pr_{(x,y)\sim s}\sqbrac{(x,y,x+y)\in \Win_i} \leq \val(\mc G) <1. \]
\end{lemma}
\begin{proof}
	Let $x_0,y_0,w \in \F_2^n$ be such that $s = s_{x_0,y_0,w}$, and $w_i\not=0$.
	We may assume that $(x_0)_i=0$ (by possibly replacing $x_0$ with $x_0+w$) and similarly $(y_0)_i=0$.
	Define the points $x_1 = x_0+w,\ y_1=y_0+w,\ z_0=x_0+y_0,\ z_1=z_0+w$.
	Also, define
	\begin{align*}
		\tilde{s} &= \set{(x,y,x+y): (x,y)\in s}
		\\&= \set{(x_0,y_0,z_0), (x_0,y_1,z_1), (x_1,y_0,z_1),(x_1,y_1,z_0)}
		\\&= \set{(x_{\alpha}, y_{\beta}, z_{\gamma}):(\alpha,\beta,\gamma)\in \supp(Q)}
	\end{align*}
	Note that $\tilde{s}\subseteq \supp(Q)^n$ and all points in $\tilde{s}$ are valid inputs to the players in the game $\mc G^{\otimes n}$.
	Moreover, these satisfy $(x_{\alpha}, y_{\beta}, z_{\gamma})_i = (\alpha,\beta,\gamma)$ for each $(\alpha,\beta,\gamma)\in \supp(Q)$.
	
	Observe that the distribution of $(x,y,x+y)$ for $(x,y)\sim s$ is simply the uniform distribution over the 4 points in $\tilde{s}$. Hence, it suffices to show that no strategies $f:\F_2^n\to \mc A,\ g:\F_2^n\to \mc B,\ h:\F_2^n\to \mc C$ for the $i$\textsuperscript{th} coordinate in $\mc G^{\otimes n}$ can win on more than $\val(\mc G)$ fraction of the points $\tilde{s}$.
	For this, consider any such functions $f,g,h$.
	Using these, we construct a strategy for the base game $\mc G$ as follows:
	\begin{enumerate}
		\item The verifier chooses $(\alpha,\beta,\gamma)\sim Q$, and gives them to the three players respectively.
		\item Player 1 outputs $f(x_{\alpha})$, Player 2 outputs $g(y_{\beta})$, and Player 3 outputs $h(z_{\gamma})$.
	\end{enumerate}
	On any input $(\alpha,\beta,\gamma)\in \supp(Q)$, as observed before, the $i$\textsuperscript{th} coordinate of the vector $(x_{\alpha}, y_{\beta}, z_{\gamma})\in \tilde{s}$ equals $(\alpha,\beta,\gamma)$.
	Hence, probability that the above strategy wins the game $\mc G$ equals the fraction of points in $\tilde{s}$ on which $f,g,h$ win the $i$\textsuperscript{th} coordinate of $\mc G^{\otimes n}$.
	By the definition of game value, this is at most $\val(\mc G)$.	
\end{proof}

Next, we show that the spreadness of $E,F,G \subseteq \mc V$ guarantees that $S = S(E, F, G)$ (see Definition~\ref{defn:diag_prod}) is covered uniformly by squares.
For this, let $\mc T$ denote the set of all squares contained in $S$; formally,
	\[ \mc T = \set{(x,y,w)\in \mc V^3 : s_{x,y,w}\subseteq S}. \]
Let $U_S$ denote the uniform distribution over $S$.
Also, let $\mu$ be the distribution on $S$ obtained as follows: Pick a random square in $S$, by picking uniformly at random $(x,y,w)\sim \mc T$; then, output a uniformly random element from the set $s_{x,y,w}$.

\begin{lemma}\label{lemma:spread_ext_sq_cover}
	We have that
	\begin{enumerate}
		\item $\abs{ \frac{\abs{S}}{\abs{\mc V}^2}-\alpha_E\alpha_F\alpha_G} \leq \frac{\epsilon}{2}\cdot \alpha_E\alpha_F\alpha_G$. 
		\item $\abs{ \frac{\abs{\mc T}}{\abs{\mc V}^3}-\alpha_E^2\alpha_F^2\alpha_G^2} \leq \frac{\epsilon}{2} \cdot \alpha_E^2\alpha_F^2\alpha_G^2$. 
		\item $\Norm{\mu - U_S}_1 \leq \frac{\epsilon}{2}$.
	\end{enumerate}
\end{lemma}
\begin{proof}
	This follows directly by Proposition~\ref{prop:sq_cover}, applied with respect to $\mc V \cong \F_2^{\dim(\mc V)}$.
	Note that the definition of squares is \emph{coordinate-free}, and hence the proposition is applicable.
\end{proof}

Before completing the proof of Lemma~\ref{lemma:hard_coor_assuming_spread}, we show that a random square inside $S=S(E, F, G)$ contains many non-trivial coordinates.
\begin{lemma}\label{lemma:many_non_triv_coord}
	Let $(x,y,w)\sim \mc T$, and let $s = s_{x,y,w}$.
	Then,
	\[ \E\abs{\mc I_{s}} \geq \frac{n}{2}\cdot (1-o(1)). \]
\end{lemma}
\begin{proof}
	For any $0<\theta<1$, we have
	\begin{align*}
		\Pr\sqbrac{\abs{\mc I_{s}} \leq (1-\theta)\cdot \frac{n}{2}} &= \Pr_{(x,y,w)\sim \mc T}\sqbrac{ \abs{\set{i:w_i\not=0}} \leq   (1-\theta)\cdot n/2}
		\\&= \frac{1}{\abs{\mc T} }\sum_{(x,y,w)\in \mc T}\ind \sqbrac{ \abs{\set{i:w_i\not=0}} \leq  (1-\theta)\cdot n/2}
		\\&\leq \frac{1}{\abs{\mc T} }\sum_{x,y,w \in \F_2^n}\ind \sqbrac{ \abs{\set{i:w_i\not=0}} \leq  (1-\theta)\cdot n/2}
		\\&= \frac{2^{3n}}{\abs{\mc T} }\Pr_{w\sim\F_2^n} \sqbrac{ \abs{\set{i:w_i\not=0}} \leq  (1-\theta)\cdot n/2}.
	\end{align*}
	By Fact~\ref{fact:chernoff}, we have, $\Pr_{w\sim\F_2^n} \sqbrac{ \abs{\set{i:w_i\not=0}} \leq  (1-\theta)\cdot n/2} \leq 2^{-\theta^2 n/4}$, and by Lemma~\ref{lemma:spread_ext_sq_cover}, we have 
	\[ \frac{\abs{\mc T}}{2^{3n}} \geq \frac{\abs{\mc V}^3}{2^{3n}} \cdot (1-\epsilon/2) \cdot \alpha_E^2\alpha_F^2\alpha_G^2 \geq 2^{-3\cdot (n-\dim(\mc V))}\cdot \frac{1}{2}\cdot 2^{-6d} \geq 2^{-\kappa n}, \]
	for some $\kappa = o(1)$.
	Combining everything, we get 
	\[ \Pr\sqbrac{\abs{\mc I_{s}} \leq (1-\theta)\cdot \frac{n}{2}} \leq 2^{\kappa n}\cdot 2^{-\theta^2 n/4} = 2^{-(\theta^2-4\kappa)n/4}. \] 
	Now, we choose $\theta = \max\set{{3\sqrt{\kappa}},n^{-0.1}} \leq o(1)$, and so the above probability is $o(1)$.
	Hence,
	\[ \E\abs{\mc I_{s}} \geq \frac{n}{2}\cdot (1-\theta)\cdot  \brac{1-o(1)} \geq \frac{n}{2}\cdot (1-o(1)). \qedhere\]
\end{proof}

Finally, we complete the proof.
\begin{proof}[Proof of Lemma~\ref{lemma:hard_coor_assuming_spread}]
	For any $i\in [n]$, and $S = S(E,F,G)$, we can write
	\begin{align*}
		\Pr_{Q^{\otimes n}}[\Win_i \mid E\times F\times G] &= \Pr_{x,y\sim\F_2^n}\sqbrac{(x,y,x+y)\in \Win_i \mid (x,y,x+y)\in E\times F\times G}
		\\&= \Pr_{x,y\sim\F_2^n}\sqbrac{(x,y,x+y)\in \Win_i \mid (x,y)\in S}
		\\&= \Pr_{(x,y)\sim S}\sqbrac{(x,y,x+y)\in \Win_i}.
	\end{align*}
	By Lemma~\ref{lemma:spread_ext_sq_cover}, this can be bounded as
	\begin{equation}\label{eqn:win_prob}
		\Pr_{Q^{\otimes n}}[\Win_i \mid E\times F\times G] \leq \E_{(x_0,y_0,w)\sim \mc T}\E_{(x,y)\sim s_{x_0,y_0,w}}\sqbrac{\ind\sqbrac{(x,y,x+y)\in \Win_i}}+\frac{\epsilon}{2}.
	\end{equation}
	Observe that for a fixed square $(x_0,y_0,w)\in \mc T,\ s = s_{x_0,y_0,w}$, by Lemma~\ref{lemma:non_triv_in_extens_hard}, we can bound
	\[\E_{(x,y)\sim s_{x_0,y_0,w}}\sqbrac{\ind\sqbrac{(x,y,x+y)\in \Win_i}} \leq 1-(1-\val(\mc G))\cdot \ind\sqbrac{i\in \mc I_s}.\]
	Plugging this into the above, and taking expectation over $i\sim [n]$, we obtain
	\begin{align*}
		\E_{i\sim [n]}\sqbrac{\Pr_{Q^{\otimes n}}[\Win_i \mid E\times F\times G]}&\leq  \E_{i\sim [n]} \E_{(x_0,y_0,w)\sim \mc T}\sqbrac{1-\brac{1-\val(\mc G)}\cdot \ind\sqbrac{i\in \mc I_{s_{x_0,y_0,w}}}} + \frac{\epsilon}{2}
		\\&= 1-\frac{\brac{1-\val(\mc G)}}{n}\E_{(x_0,y_0,w)\sim \mc T}|\mc I_{s_{x_0,y_0,w}}| + \frac{\epsilon}{2}.
		\\&\leq 1-\brac{1-\val(\mc G)}\cdot \frac{1}{2}\cdot (1-o(1)) + \frac{\epsilon}{2} \qquad (\text{by Lemma~\ref{lemma:many_non_triv_coord}})
		\\&\leq \frac{1}{2}+\frac{\val(\mc G)}{2}+\frac{\epsilon}{2}+o(1) \leq \frac{1}{2}+\frac{\val(\mc G)}{2}+\epsilon. \qedhere
	\end{align*}
\end{proof}

\subsection{Hard Coordinate under General Product Events}\label{sec:hard_coor}

Now, we work with general product events.
Consider the game $\mc G^{\otimes n} = (\mc X^n\times \mc Y^n\times \mc Z^n,\ \mc A^n\times \mc B^n\times \mc C^n,\ Q^{\otimes n}, V_{pred}^{\otimes n})$ for some sufficiently large $n\in \N$.
Fix any strategies for the $3$ players in this game, and for each $i\in [n]$, let $\Win_i\subseteq \mc X^n\times \mc Y^n\times \mc Z^n$ be the event that this strategy wins the $i$\textsuperscript{th} coordinate of the game.
We prove that:

\begin{lemma}\label{lemma:hard_coor_general}
	For any constant $\epsilon \in (0,1/4)$, there exists a constant $c=c(\epsilon)>0$, such that the following holds:
	For any sets $E,F,G\subseteq \F_2^n$ with $\Pr_{Q^{\otimes n}}\sqbrac{E\times F\times G} \geq 2^{-n^c}$, it holds that
	\[ \E_{i\sim [n]}\sqbrac{\Pr_{Q^{\otimes n}}[\Win_i \mid E\times F\times G]} \leq \frac{1}{2}+\frac{\val(\mc G)}{2}+\epsilon.\]
\end{lemma}
\begin{proof}
	Let $d = n^{1/20}$, and let $r = O_{\epsilon}(d^{16}) \in \N$ and constant $0<\delta = \delta(\epsilon)<1/10$ satisfy the statement of Lemma~\ref{lemma:hard_coor_assuming_spread}, with parameters $d,\epsilon/2$.

	Let $\eta = \epsilon/3$, and let $C=C_{\eta}$ be the constant as in Proposition~\ref{prop:uniformization}.
	Let $d_0 = d^{1/C} = n^{1/(20C)}$ and let $E,F,G\subseteq \F_2^n$ be such that $\alpha= \Pr_{Q^{\otimes n}}\sqbrac{E\times F\times G} \geq 4\cdot 2^{-d_0}$.
	Note that the lemma statement will hold with $c(\epsilon) = 1/(40C)$.
	
	Let $S = S(E,F,G)$ be as in Definition~\ref{defn:diag_prod}; this satisfies $\abs{S} = 2^{2n}\alpha$.
	Now, by Proposition~\ref{prop:uniformization}, (with parameters $r,\delta,\eta$, and the space $\F_2^n$),  we find an integer $T\in \N$, and for each $t\in [T]$, a linear subspace $\mc V_t\subseteq \F_2^n$, points $e_t,f_t\in \mc \F_2^n$, and subsets $E_t\subseteq e_t+\mc V_t,\ F_t\subseteq f_t+\mc V_t,\ G_t\subseteq e_t+f_t+\mc V_t$, such that:
	\begin{enumerate}
		\item For all $t\in [T]$, we have $\dim(\mc V_t) \geq n-  r\delta^{-3}\log_2(4/\alpha)^C = n-O(d^{17}) \geq n-o(n).$
		
		\item $S(E_1,F_1,G_1),\dots,S(E_T,F_T,G_T)$ are disjoint subsets of $S$ such that \[\abs{S\setminus \cup_{t=1}^TS(E_t,F_t,G_t)} \leq \eta \abs{S} = \eta\alpha\cdot 2^{2n}.\]
		
		\item For all $t\in [T]$, we have $\abs{E_t},\abs{F_t},\abs{G_t} \geq 2^{-\log_2(4/\alpha)^C}\abs{\mc V_t} \geq 2^{-d}\abs{\mc V_t}$, and $E_t,F_t,G_t$ are $(r,\delta)$-algebraically spread within $e_t+\mc V_t,f_t+\mc V_t,e+f_t+\mc V_t$ respectively.
	\end{enumerate}
	
	Now, for any $i\in [n]$, we may write
	\begin{align*}
	 \Pr_{Q^{\otimes n}}[\Win_i \mid E\times F\times G]  &= \alpha^{-1}\cdot \Pr_{Q^{\otimes n}}\sqbrac{\Win_i \cap E\times F\times G}
	 \\&= \alpha^{-1}\cdot \Pr_{x,y\sim \F_2^n}\sqbrac{(x,y,x+y)\in \Win_i \cap (x,y)\in S(E,F,G)}
	 \\&\leq \eta+ \sum_{t=1}^T \alpha^{-1}\cdot \Pr_{x,y\sim \F_2^n}\sqbrac{(x,y,x+y)\in \Win_i \cap (x,y)\in S(E_t,F_t,G_t)} 
	 \\&= \eta+ \sum_{t=1}^T \alpha^{-1}\cdot \Pr_{Q^{\otimes n}}\sqbrac{\Win_i \cap E_t\times F_t\times G_t} 
	 \\&= \eta+ \sum_{t=1}^T \alpha^{-1}\cdot \Pr_{Q^{\otimes n}}\sqbrac{\Win_i \mid E_t\times F_t\times G_t} \cdot \frac{\abs{S(E_t,F_t,G_t)}}{2^{2n}}.
	\end{align*}
	Also, by Lemma~\ref{lemma:hard_coor_assuming_spread}, we have for each $t\in [T]$, that
	\[ \E_{i\sim [n]}\sqbrac{\Pr_{Q^{\otimes n}}[\Win_i \mid E_t\times F_t\times G_t]} \leq \frac{1+\val(\mc G)+\epsilon}{2}. \]
	Here, we note that when the inputs are drawn  conditioned on $E_t\times F_t\times G_t$, the three players can shift them by $e_t,\ f_t,\ e_t+f_t\in \F_2^n$ respectively; now the inputs to each player lie in spread subsets of the linear subspace $\mc V_t$, and the lemma is applicable.
	Finally, combining the above two equations, and recalling $\eta=\epsilon/3$,  we get
	\begin{align*}
	 \E_{i\sim [n]}\sqbrac{\Pr_{Q^{\otimes n}}[\Win_i \mid E\times F\times G]}  &\leq \eta+ \sum_{t=1}^T \alpha^{-1}\cdot \E_{i\sim [n]}\sqbrac{\Pr_{Q^{\otimes n}}\sqbrac{\Win_i \mid E_t\times F_t\times G_t}} \cdot \frac{\abs{S(E_t,F_t,G_t)}}{2^{2n}}
	 \\&\leq \eta +\sum_{t=1}^T \alpha^{-1}\cdot \brac{\frac{1+\val(\mc G)+\epsilon}{2}}\cdot \frac{\abs{S(E_t,F_t,G_t)}}{2^{2n}}
	 \\&= \eta + \alpha^{-1}\cdot \brac{\frac{1+\val(\mc G)+\epsilon}{2}} \cdot \frac{\sum_{t=1}^T\abs{S(E_t,F_t,G_t)}}{2^{2n}}
	 \\&\leq \eta + \alpha^{-1}\cdot \brac{\frac{1+\val(\mc G)+\epsilon}{2}} \cdot \frac{\abs{S}}{2^{2n}} 
	 =  \eta+\frac{1+\val(\mc G)+\epsilon}{2} 
	 \\&\leq \frac{1}{2}+\frac{\val(\mc G)}{2}+\epsilon. \qedhere
	\end{align*}
\end{proof}


With the above, the proof of our main theorem is direct:
\begin{proof}[Proof of Theorem~\ref{thm:ghz_unif_parrep}]
	The theorem follows by combining Lemma~\ref{lemma:inductive_criteria} with  Lemma~\ref{lemma:hard_coor_general}, applied with the parameter choice $\epsilon=0.1$, recalling that $\val(\mc G)\leq 3/4$. 
\end{proof}

\begin{proof}[Proof of Theorem~\ref{thm:intro_main}]
	The theorem follows by combining Theorem~\ref{thm:ghz_unif_parrep} and Lemma~\ref{lemma:wlog_unif_dist}. 
\end{proof}


\section{A Concentration Bound}\label{sec:ghz_conc}

We also prove a concentration bound for the GHZ game:

\begin{theorem}\label{thm:ghz_conc}
	Let $\mc G = (\mc X\times \mc Y\times \mc Z,\ \mc A\times \mc B\times \mc C,\ Q, V_{pred})$ be any 3-player game with $\mc X = \mc Y = \mc Z = \F_2$, and with $Q$ the uniform distribution over 
	\[ \supp(Q) =  \set{(x,y,z)\in \F_2^3:x+y+z=0},\]
	and such that $\val(\mc G)<1$.
	Then, for every constant $\epsilon\in (0,1/4)$, there exists a constant $c=c(\epsilon)>0$, such that the following holds for all sufficiently large $n$:\footref{footnote:suff_large_n_eps}
	
	Consider the game $\mc G^{\otimes n}$, and consider any strategies for the 3 players in this game.
	Then, if $Z$ denotes the number of coordinates that the players win, it holds that
	\[ \Pr\sqbrac{Z \geq (\val(\mc G)+\epsilon)\cdot n} \leq \exp\brac{-n^{c}}.\]
\end{theorem}

The remainder of this section is devoted to the proof of the above theorem.
Towards this, we fix a game $\mc G$ as in the statement of the theorem.
Consider the game $\mc G^{\otimes n} = (\mc X^n\times \mc Y^n\times \mc Z^n,\ \mc A^n\times \mc B^n\times \mc C^n,\ Q^{\otimes n}, V_{pred}^{\otimes n})$ for some sufficiently large $n\in \N$.
Also, fix any strategies for the $3$ players in this game, and for each $i\in [n]$, let $\Win_i\subseteq \mc X^n\times \mc Y^n\times \mc Z^n$ be the event that this strategy wins the $i$\textsuperscript{th} coordinate of the game.
For each $S\subseteq [n]$, we use $\Win_S$ to denote the event $\land_{i\in S}\Win_i$.


\subsection{An Improved Bound for Product Events}

The first step in our proof is to prove an improved version of Lemma~\ref{lemma:hard_coor_general}, where we don't lose a ``factor of $2$.''
We prove the following:

\begin{lemma}\label{lemma:improved_hard_coor_general}
	For any constant $\epsilon \in (0,1/4)$, there exists $c=c(\epsilon)>0$, such that the following holds:
	For any sets $E,F,G\subseteq \F_2^n$ with $\Pr_{Q^{\otimes n}}\sqbrac{E\times F\times G} \geq 2^{-n^c}$, it holds that
	\[ \E_{i\sim [n]}\sqbrac{\Pr_{Q^{\otimes n}}[\Win_i \mid E\times F\times G]} \leq \val(\mc G)+\epsilon.\]
\end{lemma}

The proof of the above lemma exactly follows the proof in Section~\ref{sec:ghz_parrep}, so we only mention how to improve it.
We note that the only place where we lost a factor of 2 in the bound was in the proof of Lemma~\ref{lemma:hard_coor_assuming_spread}; the proof of Lemma~\ref{lemma:hard_coor_general}, using uniformization, goes through as is if an improved version of Lemma~\ref{lemma:hard_coor_assuming_spread} is proved.
Now, in Lemma~\ref{lemma:hard_coor_assuming_spread}, we proceed in exactly the same manner upto Equation~\eqref{eqn:win_prob}, which says (after taking expectation over $i\sim [n]$)
\[ \E_{i\sim [n]}\sqbrac{\Pr_{Q^{\otimes n}}[\Win_i \mid E\times F\times G]} \leq \E_{i\sim [n]}\E_{(x_0,y_0,w)\sim \mc T}\E_{(x,y)\sim s_{x_0,y_0,w}}\sqbrac{\ind\sqbrac{(x,y,x+y)\in \Win_i}}+\frac{\epsilon}{2}.  \]
The factor 2 loss now occurred because when $(x_0,y_0,w)\sim \mc T$, the event $w_i=1$ occurs with probability roughly 1/2 (see Lemma~\ref{lemma:many_non_triv_coord}), and we can only apply Lemma~\ref{lemma:non_triv_in_extens_hard} when this  holds.
To improve upon this bound, we will replace the distribution with one where we are already conditioning on the event $w_i=1$, and only then apply Lemma~\ref{lemma:non_triv_in_extens_hard}; this requires proving a stronger version of Lemma~\ref{lemma:many_non_triv_coord}, which says that the distribution of $(x,y)$ in the above statement remains unaffected even when the square is sampled conditioned on $w_i=1$.
Formally, we prove:

\begin{lemma}\label{lemma:distribution_cond_non_trivial}
	Recall the distribution $\mu$ on $S$ obtained as follows: choose $(x_0,y_0,w)\sim \mc T$, and let $(x,y)\sim s_{x_0,y_0,w}$.
	For any $i\in [n]$, define the distribution $\nu_i$ as follows: choose $(x_0,y_0,w) \sim \mc T$ conditioned on $w_i=1$, and let $(x,y)\sim s_{x_0,y_0,w}$.
	Then, we have
	\[ \E_{i\sim [n]}\Norm{\mu-\nu_i}_1 \leq o(1). \]
\end{lemma}

In the proof, we will need the following technical lemma, whose proof is deferred to Appendix~\ref{sec:app_marginal}.

\begin{lemma}\label{lemma:cond_marginal}
	Let $\mc T \subseteq \F_2^n\times\F_2^n\times \F_2^n$ be any set of size at least $2^{3n-o(n)}$.
	Let $(X,Y,W)\sim \mc T$ be chosen uniformly at random.
	Then,	\[ \E_{i\sim [n]}\Norm{P_{X,Y|W_i=1}-P_{X,Y}}_1 \leq o(1). \]
    Here, $P_{X,Y}$ denotes the marginal distribution of $(X,Y)$, and $P_{X,Y|W_i=1}$ denotes the marginal distribution of $(X,Y)$ conditioned on the event $W_i=1$.
\end{lemma}

\begin{proof}[Proof of Lemma~\ref{lemma:distribution_cond_non_trivial}]
	First, observe the distribution $\mu$ is the same as the following distribution: choose $(x_0,y_0,w)\sim \mc T$, and output $(x_0,y_0)$.
	This can seen from Equation~\eqref{eqn:square_count}: essentially, this follows by Remark~\ref{remark:sq_rep}, as every square with $w\not=0$ occurs 4 times in $\mc T$, and each point in such a square is chosen with probability $1/4$ after the square is chosen; also, each square with $w=0$ occurs only once, however the point within this square is chosen with probability $1$ after such a square is chosen.
	By a similar reasoning, for every $i\in [n]$, the distribution $\nu_i$ is the same as the following: choose $(x_0,y_0,w) \sim \mc T \mid w_i=1$, and output $(x_0,y_0)$.
	
	Now, observe that by Lemma~\ref{lemma:spread_ext_sq_cover}, we have
	\[ \abs{\mc T} \geq \abs{\mc V}^3 \cdot (1-\epsilon/2)\cdot \alpha_E^2\alpha_F^2\alpha_G^2 \geq 2^{3n-(3(n-\dim(\mc V))+1+6d)} \geq 2^{3n-o(n)}.\]
	The result now follows by Lemma~\ref{lemma:cond_marginal}.
\end{proof}

Finally, we complete the proof of the improved bound:
\begin{proof}[Proof of Lemma~\ref{lemma:improved_hard_coor_general}]
	As observed before, it suffices to improve Lemma~\ref{lemma:hard_coor_assuming_spread}, and we have
	\[ \E_{i\sim [n]}\sqbrac{\Pr_{Q^{\otimes n}}[\Win_i \mid E\times F\times G]} \leq \E_{i\sim [n]}\E_{(x_0,y_0,w)\sim \mc T}\E_{(x,y)\sim s_{x_0,y_0,w}}\sqbrac{\ind\sqbrac{(x,y,x+y)\in \Win_i}}+\frac{\epsilon}{2}.  \]
	By Lemma~\ref{lemma:distribution_cond_non_trivial} and Lemma~\ref{lemma:non_triv_in_extens_hard}, the above is at most 
	\begin{align*}
		 \E_{i\sim [n]}\E_{(x_0,y_0,w)\sim \mc T \mid w_i=1}\E_{(x,y)\sim s_{x_0,y_0,w}}\sqbrac{\ind\sqbrac{(x,y,x+y)\in \Win_i}}+\frac{\epsilon}{2} + o(1) 
		 &\leq  \val(\mc G)+\frac{\epsilon}{2} + o(1)
		 \\&\leq \val(\mc G)+\epsilon. \qedhere
	\end{align*}
\end{proof}


\subsection{Random Sets are Hard under Parallel Repetition}

We prove a strengthening of Theorem~\ref{thm:ghz_unif_parrep}, by using an argument similar to the proof of Lemma~\ref{lemma:inductive_criteria}.

\begin{lemma}\label{lemma:random_set_hard}
	For every constant $\epsilon\in (0,1/4)$, there exists $c=c(\epsilon)>0$, such that for every integer $1\leq t\leq n^c$, it holds that
	\[ \E_{|S|=t} \sqbrac{\Pr_{Q^{\otimes n}}\sqbrac{\Win_S}} \leq 2\cdot \brac{\val(\mc G)+\epsilon}^t ,\]
	where $\E_{|S|=t}$ denotes that the expectation is over a uniformly chosen subset $S\subseteq [n]$ of size $t$.
\end{lemma}
\begin{proof}
	Let $\epsilon\in (0,1/4)$ be any constant, and let $0<c\leq 1$ be as in Lemma~\ref{lemma:improved_hard_coor_general} with respect to the parameter $\epsilon/2$.
	We show that the lemma holds with constant $c/2$.

	More formally, by induction on $t$, we show that for every $1\leq t \leq \lfloor n^{c/2}\rfloor$, 
	\[\E_{|S|=t} \sqbrac{\Pr_{Q^{\otimes n}}\sqbrac{\Win_S}} \leq \brac{\val(\mc G)+\epsilon}^t + t\cdot 2^{-n^c/2} \leq 2\cdot (\val(\mc G)+\epsilon)^t.\]
	Note that the second inequality holds as we may assume that $\val(\mc G)\geq 1/4$ (or else $\val(\mc G)=0$ and there is nothing to prove) in which case $t\cdot 2^{-n^c/2} \leq (1/4)^{n^{c/2}} \leq (1/4)^t \leq (\val(\mc G)+\epsilon)^t$.
	
	The base case $t=1$ holds trivially as $\Pr_{Q^{\otimes n}}[\Win_i]\leq \val(\mc G)$ for each $i\in [n]$.
	For the inductive step, consider any $1\leq t< \lfloor n^{c/2}\rfloor$, and suppose that $\E_{|T|=t} \sqbrac{\Pr_{Q^{\otimes n}}\sqbrac{\Win_T}} \leq \brac{\val(\mc G)+\epsilon}^t+t\cdot 2^{-n^c/2}$.
	Then, we have
	\begin{align*}
		 \E_{|S|=t+1} \sqbrac{\Pr_{Q^{\otimes n}}\sqbrac{\Win_S}} 
		 &= \E_{|T|=t}\E_{i\sim [n]\setminus T} \sqbrac{\Pr_{Q^{\otimes n}}\sqbrac{\Win_{T}\land \Win_i}} 
		 \\&= \E_{|T|=t}\sqbrac{\Pr_{Q^{\otimes n}}\sqbrac{\Win_{T}}\cdot \E_{i\sim [n]\setminus T} \sqbrac{\Pr_{Q^{\otimes n}}\sqbrac{\Win_i \mid \Win_T}}}.
	\end{align*}
	With the above expression in mind, we fix any set $T\subseteq [n]$ of size $t$, and calculate an upper bound for $\E_{i\sim [n]\setminus T} \sqbrac{\Pr\sqbrac{\Win_i \mid \Win_T}}$.
	Let $R \in (\F_2\times \F_2\times \F_2 \times \mc A\times \mc B\times \mc C)^T$  be the random variable consisting of the questions and answers of the players in coordinates given by set $T$; note that this random variable takes at most $(8|\mc A||\mc B||\mc C|)^t \leq 2^{O(t)}$ possible values.
    Also, observe that the event $\Win_T$ is a deterministic function of $R$.
	Let $\mc R$ be the set of all values $r$ of the random variable $R$ that satisfy the event $\Win_T$, and let $\mc R' \subseteq \mc R$ be the set of all $r\in \mc R$ such that $\Pr_{Q^{\otimes n}}[R=r] \geq 2^{-n^c}$.
	Note that by Lemma~\ref{lemma:improved_hard_coor_general}, we have that for each $r\in \mc R'$,  \[ \E_{i\sim [n]}\Pr_{Q^{\otimes n}}\sqbrac{\Win_i\mid R=r} \leq \val(\mc G)+\frac{\epsilon}{2}, \]
	since the event $R=r$ is a product event with respect to the three players.
    
	Hence, we get
	\begin{align*}
		\E_{i\sim [n]\setminus T}\sqbrac{\Pr_{Q^{\otimes n}}\sqbrac{\Win_i\mid \Win_T}} &\leq \E_{i\sim [n]}\sqbrac{\Pr_{Q^{\otimes n}}\sqbrac{\Win_i\mid \Win_T}} + \frac{2t}{n} \qquad \text{(Lemma~\ref{lemma:large_subset_random})}
		\\&= \sum_{r\in \mc R} \frac{\E_{i\sim [n]}\sqbrac{\Pr_{Q^{\otimes n}}\sqbrac{\Win_i\mid R=r}}\cdot \Pr_{Q^{\otimes n}}[R=r]}{\Pr_{Q^{\otimes n}}[\Win_T]} + o(1)
		\\&\leq  \sum_{r\in \mc R'}\brac{\val(\mc G)+\frac{\epsilon}{2}}\cdot \frac{\Pr_{Q^{\otimes n}}[R=r]}{\Pr_{Q^{\otimes n}}[\Win_T]} +\sum_{r\in \mc R\setminus \mc R'} \frac{1\cdot 2^{-n^c}}{\Pr_{Q^{\otimes n}}[\Win_T]}+ o(1)
		\\&\leq \brac{\val(\mc G)+\frac{\epsilon}{2}} + \frac{(8|\mc A||\mc B||\mc C|)^t\cdot 2^{-n^c}}{\Pr_{Q^{\otimes n}}[\Win_T]} +o(1)
		\\&\leq \val(\mc G)+\epsilon + \frac{2^{O(n^{c/2})-n^c}}{\Pr_{Q^{\otimes n}}[\Win_T]} \leq \val(\mc G) + \epsilon+ \frac{2^{-n^c/2}}{\Pr_{Q^{\otimes n}}[\Win_T]}. 
	\end{align*}

	Combining this with the above, and using the inductive hypothesis, we get
	\begin{align*}
		 \E_{|S|=t+1} \sqbrac{\Pr_{Q^{\otimes n}}\sqbrac{\Win_S}} 
		 &= \E_{|T|=t}\sqbrac{\Pr_{Q^{\otimes n}}\sqbrac{\Win_{T}}\cdot \brac{\val(\mc G) + \epsilon+ \frac{2^{-n^c/2}}{\Pr_{Q^{\otimes n}}[\Win_T]}}}.
		 \\&\leq \E_{|T|=t}\sqbrac{\Pr_{Q^{\otimes n}}\sqbrac{\Win_{T}}}\cdot \brac{\val(\mc G) + \epsilon} + 2^{-n^{c}/2}.
		 \\&\leq \brac{\brac{\val(\mc G)+\epsilon}^t + t\cdot 2^{-n^c/2} }\cdot \brac{\val(\mc G) + \epsilon} + 2^{-n^{c}/2}
		 \\&\leq \brac{\val(\mc G)+\epsilon}^{t+1} + (t+1)\cdot 2^{-n^c/2}. \qedhere 
	\end{align*}
\end{proof}


\subsection{Proof of the Concentration Bound}

Finally, we complete the proof of our concentration bound, using standard arguments~\cite{Rao11}:

\begin{proof}[Proof of Theorem~\ref{thm:ghz_conc}]
	Consider any constant $\epsilon\in (0,1/4)$, let $0<c< 1$ be as in Lemma~\ref{lemma:random_set_hard} with the parameter $\epsilon/2$, and let $t = \lfloor n^c\rfloor$.
	
	Let $Z$ be the random variable denoting the number of coordinates won by the players.
	Whenever $Z\geq (\val(\mc G)+\epsilon)\cdot n$, we pick a uniformly random subset $S$ of size $t$ from the coordinates that the players won.
	Note that for any fixed subset $T\subseteq [n]$ of size $t$, the probability that $S$ equals $T$ is at most $\binom{(\val(\mc G)+\epsilon)\cdot n}{t}^{-1}$. 
    Hence, we have 
    \begin{align*}
        \Pr\sqbrac{Z\geq (\val(\mc G)+\epsilon)\cdot n}
		&\leq \sum_{T\subseteq [n], |T|=t} \Pr[S=T] \cdot \Pr[\Win_T]
		\\&\leq \sum_{T\subseteq [n], |T|=t} \binom{(\val(\mc G)+\epsilon)\cdot n}{t}^{-1} \cdot \Pr[\Win_T].
    \end{align*}
    Now, using Lemma~\ref{lemma:random_set_hard}, we get
    
	\begin{align*}
		\Pr\sqbrac{Z\geq (\val(\mc G)+\epsilon)\cdot n} &\leq \binom{n}{t}\cdot  \binom{(\val(\mc G)+\epsilon)\cdot n}{t}^{-1}\cdot 2\cdot\brac{\val(\mc G)+\epsilon/2}^t 
		\\&\leq \brac{\frac{n}{(\val(\mc G)+\epsilon)\cdot n-t}}^t \cdot 2\cdot\brac{\val(\mc G)+\epsilon/2}^t
        \\&\leq 2 \cdot \brac{\frac{\val(\mc G)+\epsilon/2}{\val(\mc G)+\epsilon - o(1)}}^t 
		\\&= 2\cdot \brac{1-\frac{\epsilon/2-o(1)}{\val(\mc G)+\epsilon-o(1)}}^t
		\leq 2\cdot \brac{1-\epsilon/4}^t
		\\&\leq 2^{1-\epsilon t/4} \leq 2^{-n^{c/2}}. &\qedhere
	\end{align*}
\end{proof}


\bibliographystyle{alpha}
\bibliography{main.bib}

\appendix
\section{Marginals of a Conditional Distribution}\label{sec:app_marginal}

In this section, we prove Lemma~\ref{lemma:cond_marginal}.
Our proof will use some basic information theory, and the reader is referred to~\cite{CT06} for an excellent introduction to information theory.
Recall that for a random variable $X$ over a finite set $\Omega$, its entropy is defined as \[H(X) = -\sum_{x\in \Omega}\Pr[X=x]\cdot \log_2(\Pr[X=x]).\]
 
We have the following simple fact:

\begin{lemma}\label{lemma:bin_ent_ineq}
	Let $X\in \set{0,1}$ be a binary valued random variable.
	Then,
	\[ \abs{\Pr[X=1]-\frac{1}{2}}^2\leq 1-H(X).\] 
\end{lemma}
\begin{proof}
	Let $\delta = \abs{\Pr[X=1]-\frac{1}{2}}$.
	The lemma statement is equivalent to the inequality
	 \[ \delta^2 \leq 1-H(X) = 1+\brac{\frac{1}{2}+\delta}\log_2\brac{\frac{1}{2}+\delta} + \brac{\frac{1}{2}-\delta}\log_2\brac{\frac{1}{2}-\delta},\]
	 which is true for all $\delta\in [0,1/2]$.
\end{proof}

Next, we prove the main lemma of this section:
\begin{proof}[Proof of Lemma~\ref{lemma:cond_marginal}]
	Suppose that $\abs{\mc T} \geq 2^{(3-\delta)n}$ for some $\delta=o(1)$.
	Then, we have
	\[ (3-\delta)n \leq H(X,Y,W) = H(X,Y) + H(W\mid X,Y) \leq 2n + \sum_{i=1}^n H(W_i\mid X,Y), \] 
	which implies \[\E_{i\sim [n]} \sqbrac{H(W_i\mid X,Y)} \geq 1-\delta.\]
	Define the set \[\mc G = \set{i\in [n] : H(W_i| X,Y) \geq 1-\sqrt{\delta}}.\]
	Then, by Markov's inequality, we have $\Pr_{i\sim [n]}\sqbrac{i\not\in \mc G} \leq \sqrt{\delta}$.
	 
	Now, consider any $i\in \mc G$.
	By Cauchy-Schwarz and Lemma~\ref{lemma:bin_ent_ineq}, we have
	\begin{align*}
		\brac{\E_{(x,y)\sim P_{X,Y}}\abs{\Pr\sqbrac{W_i=1 |X=x,Y=y}-\frac{1}{2}}}^2 &\leq \E_{(x,y)\sim P_{X,Y}}\abs{\Pr\sqbrac{W_i=1 | x,y}-\frac{1}{2}}^2
		\\&\leq \E_{(x,y)\sim P_{X,Y}}\sqbrac{1-H(W_i|x,y)}
		\\&= 1-H(W_i|X,Y) \leq \sqrt{\delta}.
	\end{align*}
	Hence, by the triangle inequality, we have
	\[ \abs{\Pr\sqbrac{W_i=1 }-\frac{1}{2}}\leq \E_{(x,y)\sim P_{X,Y}}\abs{\Pr\sqbrac{W_i=1 | x,y}-\frac{1}{2}} \leq \delta^{1/4}. \]
	By the above, we get
	\begin{align*}
		\Norm{P_{X,Y|W_i=1}-P_{X,Y}}_1  &= \sum_{x,y}\abs{\Pr[x,y|W_i=1]-\Pr[x,y]}
		\\&= \E_{(x,y)\sim P_{X,Y}} \abs{\frac{\Pr[W_i=1|x,y]}{\Pr[W_i=1]}-1}
		\\&\leq \E_{(x,y)\sim P_{X,Y}} \abs{\frac{\Pr[W_i=1|x,y]}{1/2}-1} \\&\qquad+ \E_{(x,y)\sim P_{X,Y}}\sqbrac{\Pr[W_i=1|x,y]}\cdot \abs{\frac{1}{\Pr[W_i=1]}-\frac{1}{1/2}}
		\\&\leq 2\cdot \delta^{1/4} +\Pr[W_i=1]\cdot \frac{\delta^{1/4}}{\Pr[W_i=1]\cdot 1/2}		
		\\&\leq 4\delta^{1/4}.
	\end{align*}
	Hence,
	\[ \E_{i\sim [n]}\Norm{P_{X,Y|W_i=1}-P_{X,Y}}_1 \leq 4\delta^{1/4} + 2\cdot \Pr_{i\sim [n]}[i\not\in \mc G] \leq 4\delta^{1/4}+2\delta^{1/2} \leq o(1). \qedhere\]
\end{proof}


\end{document}